\documentclass[envcountsame]{llncs}
\usepackage{times}
\usepackage{graphicx,enumerate}
\usepackage{mdwlist}

\usepackage[T1]{fontenc}

\usepackage{amssymb,amsmath}
\usepackage{tikz}
\usetikzlibrary{snakes}
\newtheorem{algorithm}[theorem]{Algorithm}

\newcommand{\ignore}[1]{}

\newcommand{\cP}{{\ensuremath{\mathbb P}}}
\newcommand{\cT}{{\ensuremath{\mathbb T}}}
\newcommand{\cS}{{\ensuremath{\mathbb S}}}
\newcommand{\cR}{{\ensuremath{\mathbb R}}}
\newcommand{\opleader}{\ensuremath{\mathtt{leader}}}
\newcommand{\vis}{\ensuremath{\mathtt{vis}}}
\newcommand{\low}{\ensuremath{\mathtt{cheap}}}
\newcommand{\set}{\ensuremath{\mathtt{set}}}
\newcommand{\level}{\ensuremath{\mathtt{level}}}
\newcommand{\paths}{\ensuremath{\mathtt{paths}}}
\newcommand{\interior}{\ensuremath{\mathtt{interior}}}
\newcommand{\parent}{\ensuremath{\mathtt{parent}}}
\newcommand{\meet}{\ensuremath{\mathtt{meet}}}
\newcommand{\jump}{\ensuremath{\mathtt{jump}}}
\newcommand{\lca}{\ensuremath{\mathtt{lca}}}

\newcommand{\origin}{\ensuremath{\mathtt{origin}}}
\newcommand{\toplevel}{\ensuremath{\mathtt{toplevel}}}
\newcommand{\metricstretch}{\ensuremath{\mathtt{stretch}}}

\authorrunning{M. Cygan, L. Kowalik, M. Mucha, M. Pilipczuk  and P. Sankowski} 

\author{Marek Cygan\inst{1}, Lukasz Kowalik\inst{1}, Marcin Mucha\inst{1}, \\
  Marcin Pilipczuk\inst{1} and Piotr Sankowski\inst{1,2} 
}

\institute{
Institute of Informatics, University of Warsaw, Poland
\and
Dipartimento di Informatica e Sistemistica, Sapienza - University of Rome, Italy
}

\ignore{
Marek Cygan$^{1}$, Lukasz Kowalik$^{1}$,  Marcin Mucha$^{1}$, \\
  Marcin Pilipczuk$^{1}$ and Piotr Sankowski$^{1,2}$ \\
\texttt{\{cygan,kowalik,mucha,malcin,sank\}@mimuw.edu.pl} \\
 \\
\small
$^{1}$ Institute of Informatics\\
\small University of Warsaw, Poland\\
\small
$^{2}$ Dipartimento di Informatica e Sistemistica\\
\small
    Sapienza - University of Rome, Italy
}

\title{Fast Approximation in Subspaces by Doubling Metric Decomposition
\thanks{This work was partially supported by the Polish Ministry of Science
grant N206 355636. E-mail addresses: \texttt{\{cygan,kowalik,mucha,malcin,sank\}@mimuw.edu.pl.}}
}

\begin{document}
\maketitle

\begin{abstract}
In this paper we propose and study a new complexity model for approximation algorithms. The main motivation are practical problems over large data sets that need to be solved many times for different scenarios, e.g., many multicast trees that need to be constructed for different groups of users. In our model we allow a preprocessing phase, when some information of the input graph $G=(V,E)$ is stored in a limited size data structure. Next, the data structure enables processing queries of the form ``solve problem A for an input $S\subseteq V$''. We consider problems like {\sc Steiner Forest}, {\sc Facility Location}, {\sc $k$-Median}, {\sc $k$-Center} and {\sc TSP} in the case when the graph induces a doubling metric. Our main results are data structures of near-linear size that are able to answer queries in time close to linear in $|S|$. This improves over typical worst case reuniting time of approximation algorithms in the classical setting which is $\Omega(|E|)$ independently of the query size. In most cases, our approximation guarantees are arbitrarily close to those in the classical setting. Additionally, we present the first fully dynamic algorithm for the Steiner tree problem.

\end{abstract}

\section{Introduction}
\paragraph{Motivation}
The complexity and size of the existing communication networks has grown extremely in the recent times. It is now hard to imagine that a group of users willing to communicate sets up a minimum cost communication network or a multicast tree according to an approximate solution to {\sc Steiner Tree} problem. Instead we are forced to use heuristics that are computationally more efficient but may deliver suboptimal results~\cite{multicast2,multicast}. It is easy to imagine other problems that in principle can be solved with constant approximation factors using state of art algorithms, but due to immense size of the data it is impossible in timely manner. However, in many applications the network is fixed and we need to solve the problem many times for different groups of users.

Here, we propose a completely new approach that exploits this fact to overcome the obstacles stemming from huge data sizes. It is able to efficiently deliver results that have good approximation guarantee thanks to the following two assumptions. We assume that the network can be preprocessed beforehand and that the group of users that communicates is substantially smaller than the size of the network. The preprocessing step is independent of the group of users and hence afterwards we can, for example, efficiently compute a Steiner tree for any set of users.

More formally, in  the {\sc Steiner Tree} problem the algorithm is given a weighted graph $G=(V,E)$ on $n$ vertices and is allowed some preprocessing. The results of the preprocessing step need to be stored in limited memory. Afterwards, the set $S \subseteq V$ of terminals is defined and the algorithm should generate as fast as possible a Steiner tree for $S$, i.e., a tree in $G$ of low weight which contains all vertices in $S$. Given the query set $S$ of $k$ vertices we should compute the Steiner tree $T$ in time depending only (or, mostly) on $k$.

The trivial approach to this problem is to compute the metric closure $G^*$ of $G$ and then answer each query by solving the {\sc Steiner Tree} problem on $G^*[S]$. This approach delivers results with constant approximation ratio, but requires $O(n^2)$ space of the data structure and $\tilde{O}(k^2)$ query time. Hence it is far from being practical. In this work we aim at solutions that substantially improve both of these bounds; more formally the data structure space should be close to $O(n)$, while the query time should be close to $O(k)$. Since in a typical situation probably $k=O(\log n)$, so even a  $O(k\log n)$ query time is not considered fast enough, as then $k\log n = \theta(k^2)$.
Note that the $O(n)$ bound on the structure size is very restrictive: in a way, this bound is sublinear in the sense that we are allowed neither to store the whole distance matrix, nor (if $G$ is dense) all the edges of $G$.
This models a situation when during the preprocessing one can use vast resources (e.g.,\ a huge cluster of servers), but the resources are not granted forever and when the system processes the queries the available space is much smaller.

\paragraph{New Model}
In our model, computations are divided into two stages: the preprocessing stage and the query stage.
In the preprocessing stage, the input is a weighted graph $G=(V,E)$ and we should compute our data structure in polynomial time and space.
Apart from the graph $G$ some additional, problem-specific information may be also provided.
In the query stage the algorithm is given the data structure computed in the preprocessing stage, but not $G$ itself, and a set $S$ of points of $V$ (the query --- possibly a set of pairs of points from $V$, or a weighted set of points from $V$, etc.) and computes a solution for the set $S$. The definition of ``the solution for the set $S$'' depends on the specific problem.
In this work we consider so-called metric problems, so $G$ corresponds to a
metric space $(V,d)$ where $d$ can be represented as the full distance matrix $M$.
One should keep in mind that the function $d$ cannot be quickly computed (e.g.\ in constant time) without the $\Omega(n^2)$ size matrix $M$. In particular, we assume that there is no distance oracle available in the query stage.

Hence, there are three key parameters of an algorithm within our model: the size of the data structure, the query time and the approximation ratio. Less important, but not irrelevant is the preprocessing time.
Let us note that though our model is inspired by large datasets, in this work we ignore streaming effects, external memory issues etc.

Above we have formulated the {\sc Steiner Tree} problem in our model, now we describe the remaining problems. In {\sc Steiner Forest} problem the algorithm is allowed to preprocess a weighted graph $G=(V,E)$, whereas the query is composed of the set $S \subseteq V\times V$ of pairs. The algorithm should generate the Steiner forest for $S$, i.e., a subgraph $H$ of $G$ of small weight such that each pair in $S$ is connected in $H$. In {\sc Facility Location} problem the algorithm is given in the preprocessing phase a weighted graph with facility opening costs in the nodes. We consider two variants of this problem in our model.
In the variant {\em with unrestricted facilities}, the query is a set $S\subseteq V$ of clients for which we should open facilities.
The goal is to open a subset $F\subseteq V$ of facilities, and connect each city to an open facility so that the sum of the total opening
and connection costs is minimized.
In the other variant, one with {\em restricted facilities}, the facilities that can be opened are given as a part of query (together with their opening costs).

\paragraph{Our Results}
In this paper we restrict our attention to doubling metric spaces which include growth-restricted metric spaces and constant dimensional Euclidean spaces. In other words we assume that the graph $G$ induces a doubling metric and the algorithms are given the distance matrix $G^*$ as an input or compute it at the beginning of the preprocessing phase. This restriction is often assumed in the routing setting~\cite{spanner_d1,spanner_d2} and hence it is a natural question to see how it can impact the multicast problems. Using this assumption we show that solutions with nearly optimal bounds are possible. The main result of the paper is the data structure that requires $O(n\log n)$ memory and can find a constant ratio approximate Steiner tree over a given set of size $k$ in $O(k (\log k + \log \log n))$ time. Moreover, we show data structures with essentially the same complexities for solving {\sc Steiner Forest}, both versions of {\sc Facility Location}, {\sc $k$-Median} and {\sc TSP}. The query bound is optimal, up to $\log k$ and $\log \log n$ factors, as no algorithm can answer queries in time less than linear in $k$ as it needs to read the input. For the exact approximation ratios of our algorithms refer to Sections~\ref{sec:spanner-app} and \ref{sec:unres-fl}.


All of these results are based on a new hierarchical data structure for representing a doubling metric that approximates original distances with $(1+\epsilon)$-multiplicative factor.
The concept of a hierarchical data structure for representing a doubling metric is not novel -- it originates from the work of Clarkson~\cite{clarkson} and was then used in a number of papers, in particular our data structure is based on the one due to Jia et al.~\cite{jia}.
Our main technical contribution here is adapting and extending this data structure so that for any subset $S\subset V$ a substructure corresponding to $S$ can be retrieved in $O(k(\log k+\log\log n))$ using only the information in the data structure, without a distance oracle.
The substructure is then transformed to a pseudo-spanner described above.
Note that our complexity bounds do not depend on the stretch of the metrics, unlike in many previous works (e.g.~\cite{kr-lee}).
Another original concept in our work is an application of spanners (or, more precisely, pseudo-spanners) to improve working time of approximation algorithms for metric problems.
As a result, the query times for the metric problems we consider are $O(k({\rm polylog} k+\log\log n))$.

Astonishingly, our hierarchical data structure can be used to obtain dynamic algorithms for the {\sc Steiner tree} problem. This problem attracted considerable attention~\cite{dst2,dst3,dst4,dst5} in the recent years. However, due to the hardness of the problem none of these papers has given any improvement in the running time over the static algorithms. Here, we give first fully dynamic algorithm for the problem in the case of doubling metric. Our algorithm is given a static graph and then maintains information about the Steiner tree built on a given set $X$ of nodes. It supports insertion of vertices in $O(\log^5k + \log\log n)$ time, and deletion in $O(\log^5k)$ time, where $k = |X|$.

\paragraph{Related Work}
The problems considered in this paper are related to several algorithmic topics studied extensively in recent years. Many researchers tried to answer the question whether problems in huge networks can be solved more efficiently than by processing the whole input. Nevertheless, the model proposed in this paper has never been considered before. Moreover, we believe that within the proposed framework it is possible to achieve complexities that are close to being practical. We present such results only in the case of doubling metric, but hope that the further study will extend these results to a more general setting. Our results are related to the following concepts:
\begin{itemize}
\item Universal Algorithms --- this model does not allow any processing in the query time, we allow it and get much better approximation ratios,
\item Spanners and Approximate Distance Oracles --- although a spanner of a subspace of a doubling metric can be constructed in $O(k\log k)$-time, the construction algorithm requires a distance oracle (i.e.\ the full $\Theta(n^2)$-size distance matrix).
\item Sublinear Approximation Algorithms --- here we cannot preprocess the data, allowing it we can get much better approximation ratios,
\item Dynamic Spanning Trees --- most existing results are only applicable to dynamic MST and not dynamic Steiner tree, and the ones concerning the latter work in different models than ours.
\end{itemize}
Due to space limitation of this extended abstract an extensive discussion of the related work is attached in Appendix~\ref{app:related_work} and will be included in the full version of the paper.

\section{Space partition tree}
\label{sec:partition-tree}

In this section we extend the techniques developed by Jia et al. \cite{jia}. Several statements as well
as the overall construction are similar to those given by Jia et al. However, our approach is tuned to
better suit our needs, in particular to allow for a fast subtree extraction and a spanner construction --
techniques introduced in Sections~\ref{sec:partition-tree} and \ref{sec:spanner} that are crucial for efficient approximation algorithms.

Let $(V, d)$ be a finite doubling metric space with $|V| = n$ and a doubling constant $\lambda$, i.e.,
for every $r > 0$, every ball of radius $2r$ can be covered with at most $\lambda$ balls of radius $r$. By $\metricstretch$ we denote the stretch of the metric $d$, that is, the largest distance in $V$ divided by the smallest distance.
We use space partition schemes for doubling metrics to create a partition tree. In the next two subsections, we
show that this tree can be stored in $O(n \log n)$ space, and that a subtree induced by any subset $S \subset V$ can be extracted efficiently.

Let us first briefly introduce the notion of a space partition tree, that is used in the remainder of this
paper. Precise definitions and proofs (in particular a proof of existence of such a partition tree)
can be found in Appendix~\ref{app:tree}.

The basic idea is to construct a sequence $\cS_0, \cS_1, \ldots, \cS_M$ of partitions of $V$. We require that $\cS_0 = \{\{v\}: v\in V\}$, and $\cS_M = \{V\}$, and in general the diameters of the sets in $\cS_k$ are growing exponentially in $k$. We also maintain the neighbourhood structure for each $\cS_k$, i.e.,\ we know which sets in $\cS_k$ are close to each other (this is explained in more detail later on). Notice that the partitions together with the neighbourhood structure are enough to approximate the distance between any two points $x,y$ --- one only needs to find the smallest $k$, such that the sets in $\cS_k$ containing $x$ and $y$ are close to each other (or are the same set).

There are two natural parameters in this sort of scheme. One of them is how fast the diameters of the sets grow, this is controlled by $\tau \in \cR, \tau \ge 1$ in our constructions. The faster the set diameters grow, the smaller the number of partitions is. The second parameter is how distant can the sets in a partition be to be still considered neighbours, this is controlled by a nonnegative integer $\eta$ in our constructions. The smaller this parameter is, the smaller the number of neighbours is. Manipulating these parameters allows us to decrease the space required to store the partitions, and consequently also the running time of our algorithms. However, this also comes at a price of lower quality approximation.

In what follows, each $\cS_k$ is a subpartition of $\cS_{k+1}$ for $k=0,\ldots,M-1$. That is, the elements of these partitions form a tree, denoted by $\cT$, with $\cS_0$ being the set of leaves and $\cS_M$ being the root. We say that $S \in \cS_j$ is a {\em{child}} of $S^* \in \cS_{j+1}$ in $\cT$ if $S \subset S^*$.

Let $r_0$ be smaller than the minimal distance between points in $V$ and let $r_j = \tau^jr_0$. We show (in Appendix~\ref{app:tree}) that $\cS_k$-s and $\cT$ satisfying the following properties can be constructed in polynomial time:
\ignore{
\begin{description}
\item[Exponential growth:] Every $S \in \cS_j$ is contained in a ball of radius $r_j\tau2^{-\eta}/(\tau-1)$.
\item[Small neighbourhoods:] \label{item:acquaintance-limit} For every $S \in \cS_j$, the union $\bigcup\{B_{r_j}(v): v \in S\}$ crosses at most $\lambda^{3+\eta}$ sets $S'$ from the partition $\cS_j$ --- we say that $S$ {\em{knows}} these $S'$. We also extend this notation and say that if $S$ knows $S'$, then every $v \in S$ knows $S'$.
\item[Small degrees:] \label{item:childlimit} For every $S^* \in \cS_{j+1}$ all children of $S^*$ know each other, and there are at most $\lambda^{\eta+3}$ children of $S^*$.
\item[Distance approximation:] \label{item:remstretchlimit} If $v, v^* \in V$ are different points such that $v \in S_1 \in \cS_j$, $v \in S_2 \in \cS_{j+1}$ and $v^* \in S_1^* \in \cS_j$, $v^* \in S_2^* \in \cS_{j+1}$ and $S_2$ knows $S_2^*$ but $S_1$ does not know $S_1^*$, then
$$r_j \leq d(v, v^*) < \Big( 1 + \frac{4\tau 2^{-\eta}}{\tau - 1}\Big) \tau r_j;$$
\end{description}
}

\begin{enumerate}[{\bf(1)}]
\item{\bf Exponential growth:} Every $S \in \cS_j$ is contained in a ball of radius $r_j\tau2^{-\eta}/(\tau-1)$.
\item{\bf Small neighbourhoods:} \label{item:acquaintance-limit} For every $S \in \cS_j$, the union $\bigcup\{B_{r_j}(v): v \in S\}$ crosses at most $\lambda^{3+\eta}$ sets $S'$ from the partition $\cS_j$ --- we say that $S$ {\em{knows}} these $S'$. We also extend this notation and say that if $S$ knows $S'$, then every $v \in S$ knows $S'$.
\item{\bf Small degrees:} \label{item:childlimit} For every $S^* \in \cS_{j+1}$ all children of $S^*$ know each other and, consequently, there are at most $\lambda^{\eta+3}$ children of $S^*$.
\item{\bf Distance approximation:} \label{item:remstretchlimit} If $v, v^* \in V$ are different points such that $v \in S_1 \in \cS_j$, $v \in S_2 \in \cS_{j+1}$ and $v^* \in S_1^* \in \cS_j$, $v^* \in S_2^* \in \cS_{j+1}$ and $S_2$ knows $S_2^*$ but $S_1$ does not know $S_1^*$, then
$$r_j \leq d(v, v^*) < \Big( 1 + \frac{4\tau 2^{-\eta}}{\tau - 1}\Big) \tau r_j;$$
For any $\varepsilon >0$, the $\tau$ and $\eta$ constants can be adjusted so that the last condition becomes $r_j \le d(v,v^*) \le (1+\varepsilon)r_j$ (see Remark~\ref{rem:stretchlimit}).
\end{enumerate}

\begin{remark}
We note that not all values of $\tau$ and $\eta$ make sense for our construction. We omit these additional constraints here.
\end{remark}

\subsection{The compressed tree $\hat{\cT}$ and additional information at nodes}

Let us now show how to efficiently compute and store the tree $\cT$.
Recall that the leaves of $\cT$ are one point sets and, while going up in the tree,
these sets join into bigger sets.

Note that if $S$ is an inner node of $\cT$ and it has only one child $S'$
then both nodes $S$ and $S'$ represent the same set.
Nodes $S$ and $S'$ can differ only by their sets of acquaintances, i.e.\ the sets of nodes known to them.
If these sets are equal, there is some sort of redundancy in $\cT$.
To reduce the space usage we store only a compressed version of the tree $\cT$.

Let us introduce some useful notation.
For a node $v$ of ${\cT}$ let $\set(v)$ denote the set corresponding to $v$ and
let $\level(v)$ denote the level of $v$, where leaves are at level zero.
Let $S_a$, $S_b$ be a pair of sets that know each other at level $j_{ab}$ and do not know each other at level $j_{ab}-1$.
Then the triple $(S_a, S_b, j_{ab})$ is called a {\em meeting} of $S_a$ and $S_b$ at level $j_{ab}$.

\begin{definition}[Compressed tree]
\label{def:compressed_tree}
The compressed version of $\cT$, denoted $\hat{\cT}$, is obtained from $\cT$
by replacing all maximal paths such that all inner nodes have exactly one child by a single edge.
For each node $v$ of $\hat{\cT}$ we store $\level(v)$ (the lowest level of $\set(v)$ in $\cT$)
and a list of all meetings of $\set(v)$, sorted by level.
\end{definition}

Obviously $\hat{\cT}$ has at most $2n-1$ nodes since it has exactly $n$ leaves
and each inner node has at least two children but we also have to ensure that
the total number of meetings is reasonable.

Note that the sets at nodes of $\hat{\cT}$ are pairwise distinct.
To simplify the presentation we will identify nodes and the corresponding sets.
Consider a meeting $m=(S_a, S_b, j_{ab})$.
Let $p_a$ (resp. $p_b$) denote the parent of $S_a$ (resp. $S_b$) in $\hat{\cT}$.
We say that $S_a$ is {\em responsible} for the meeting $m$ when
$\level(p_a)\le\level(p_b)$ (when $\level(p_a)=\level(p_b)$, both
$S_a$ and $S_b$ are responsible for the meeting $m$).
Note that if $S_a$ is responsible for a meeting $(S_a, S_b, j_{ab})$,
then $S_a$ knows $S_b$ at level $\level(p_a)-1$.
From this and Property~\ref{item:acquaintance-limit} of the partition tree we get the following.

\begin{lemma}
\label{lem:bounded-responsibility}
 Each set in $\hat{\cT}$ is responsible for at most $\lambda^{3+\eta}$ meetings. 
\end{lemma}

\begin{corollary}
\label{lem:num_edges}
There are $\le(2n-1)\lambda^{3 + \eta}$ meetings stored in the compressed tree $\hat{\cT}$, i.e.\ $\hat{\cT}$ takes $O(n)$ space.
\end{corollary}

\begin{lemma}
\label{lem:know-query}
One can augment the tree $\hat{\cT}$ with additional information of size $O(n\lambda^{3 + \eta})$,
so that for any pair of nodes $x, y$ of $\hat{\cT}$ one can decide if $x$ and $y$ know each other, and if that is the case the level of the meeting is returned.
The query takes $O(\eta\log\lambda)$ time.
\end{lemma}

\begin{proof}
For each node $v$ in $\hat{\cT}$ we store all the meetings it is responsible for,
using a dictionary $D(m)$ --- the searches take $O(\log(\lambda^{3+\eta}))=O(\eta\log\lambda)$ time.
To process the query it suffices to check if there is an appropriate meeting
in $D(x)$ or in $D(y)$.\qed
\end{proof}


In order to give a fast subtree extraction algorithm, we need to define the following operation $\meet$.
Let $u,v\in \hat{\cT}$ be two given nodes. Let $v(j)$ denote the node in $\cT$ on the path
from $v$ to the root at level $j$, similarly define $u(j)$.
The value of $\meet(u, v)$ is the lowest level, such that $v(j)$ and $u(j)$ know each other. Such level
always exists, because in the end all nodes merge into root and nodes know each other at one level
before they are merged (see Property~\ref{item:childlimit} of the partition tree).
A technical proof of the following lemma is moved to Appendix~\ref{app:meet} due to space limitations.

\begin{lemma}
\label{lem:binary-jumping}
The tree $\hat{\cT}$ can be augmented so that the $\meet$ operation can be
performed in $O(\eta\log\lambda \log \log n)$ time.
The augmented $\cT$ tree can be stored in $O(\lambda^{3+\eta}n\log n)$ space and computed in polynomial time.
\end{lemma}

\ignore{
\begin{lemma}
The above augmented $\cT$ tree can be stored in $O(\lambda^{3+\eta}n\log n)$ space and computed in polynomial time.
\end{lemma}

\begin{proof}
The additional $\log n$ factor in the space bound comes from the size of $\paths(x)$ for each node $x$ (See Appendix\ref{app:meet}).
We need only to describe how to obtain running time independent of the stretch of the metric.
In order to compute the $\hat{\cT}$ tree (without augmentation) we can slightly improve our construction algorithm:
instead of going into the next level, one can compute the smallest distance between current sets and jump directly
to the level when some pair of sets merges or begins to know each other.\qed
\end{proof}
}

\subsection{Fast subtree extraction}
\label{sec:fast-subtree-extraction}
For any subset $S\subseteq V$ we are going to define an {\em $S$-subtree} of $\hat{\cT}$,
denoted $\hat{\cT}(S)$.
Intuitively, this is the subtree of $\hat{\cT}$ induced by the leaves
corresponding to $S$.
Additionally we store all the meetings in $\hat{\cT}$ between the nodes corresponding to the nodes of $\hat{\cT}(S)$.

More precisely, the set of nodes of $\hat{\cT}(S)$ is defined as
$\{A \cap S\ :\ $ $A\subseteq V$ and $A$ is a node of $\hat{\cT}\}$.
A node $Q$ of $\hat{\cT}(S)$ is an ancestor of a node $R$ of $\hat{\cT}(S)$ iff $R\subseteq Q$.
This defines the edges of $\hat{\cT}(S)$.
Moreover, for two nodes $A$, $B$ of $\hat{\cT}$ such that both $A$ and $B$ intersect $S$,
if $A$ knows $B$ at level $j$, we say that $A \cap S$ {\em knows} $B \cap S$ in $\hat{\cT}(S)$ at level $j$.
A triple $(Q,R,j_{QR})$, where $j_{QR}$ is a minimal level such that $Q$ knows $R$ at level $j_{QR}$, is called a {\em meeting}.
The {\em level} of a node $Q$ of $\hat{\cT}(S)$ is the lowest level of a node $A$ of $\hat{\cT}$ such that $Q=A\cap S$.
Together with each node $Q$ of $\hat{\cT}(S)$ we store its level and
a list of all its meetings $(Q,R,j_{QR})$.
A node $Q$ is {\em responsible} for a meeting $(Q,R,l)$ when $\level(\parent(Q))\le \level(\parent(R))$.

\begin{remark}
The subtree  $\hat{\cT}(S)$ is not necessarily equal to any compressed tree for the metric space $(S,d|_{S^2})$.
\end{remark}

In this subsection we describe how to extract $\hat{\cT}(S)$ from $\hat{\cT}$ efficiently.
The extraction runs in two phases.
In the first phase we find the nodes and edges of $\hat{\cT}(S)$ and in the second phase we find the meetings.

\subsubsection{Finding the nodes and edges of $\hat{\cT}(S)$ }
We construct the extracted tree in a bottom-up fashion.
Note that we can not simply go up the tree from the leaves corresponding to $S$ because we
could visit a lot of nodes of $\hat{\cT}$ which are not the nodes of $\hat{\cT}(S)$.
The key observation is that if $A$ and $B$ are nodes of $\hat{\cT}$,
such that $A\cap S$ and $B \cap S$ are nodes of $\hat{\cT}(S)$ and $C$ is the
lowest common ancestor of $A$ and $B$, then $C\cap S$ is a node of $\hat{\cT}(S)$ and it has level
$\level(C)$.

\begin{enumerate}
  \item Sort the leaves of $\hat{\cT}$ corresponding to the elements of $S$ according to their inorder value in $\hat{\cT}$, i.e.,\ from left to right.
  \item For all pairs $(A,B)$ of neighboring nodes in the sorted order, insert into a dictionary $M$
    a key-value pair where the key is the pair $(\level(\lca_{\hat{\cT}}(A,B)), \lca_{\hat{\cT}}(A,B))$ and the value is the pair $(A,B)$.
    The dictionary $M$ may contain multiple elements with the same key.
  \item Insert all nodes from $S$ to a second dictionary $P$, where nodes are sorted according to their
    inorder value from the tree $\hat{\cT}$.

  \item while $M$ contains more than one element
    \begin{enumerate}
      \item Let $x=(l,C)$ be the smallest key in $M$.
      \item Extract from $M$ all key-value pairs with the key $x$, denote those
      values as $(A_1,B_1),\ldots,(A_m,B_m)$.
      \item Set $P = P \setminus \bigcup_i\{A_i,B_i\}$.
      \item Create a new node $Q$, make the nodes erased from $P$ the children of $Q$. Store $l$ as the level of $Q$.
      \item Insert $C$ into $P$. Set $\origin(Q)=C$.
      \item If $C$ is not the smallest element in $P$ (according to the inorder value)
            let $C_l$ be the largest element in $P$ smaller than $C$ and add
            a key-value pair to $M$ where the key is equal to $(\level(\lca_{\hat{\cT}}(C_l,C)), \lca_{\hat{\cT}}(C_l,C))$ and the value is $(C_l,C)$.
      \item If $C$ is not the largest element in $P$
            let $C_r$ be the smallest element in $P$ larger than $C$ and add
            a key-value pair to $M$ where the key is given by the pair $(\level(\lca_{\hat{\cT}}(C,C_r)), \lca_{\hat{\cT}}(C,C_r))$ and the value is the pair $(C,C_r)$.
    \end{enumerate}
\end{enumerate}

Note that in the above procedure, for each node $Q$ of $\hat{\cT}(S)$ we
compute the corresponding node in $\hat{\cT}$, namely $\origin(Q)$.
Observe that $\origin(Q)$ is the lowest common ancestor of the leaves corresponding to elements of $Q$, and $\origin(Q)\cap S=Q$.

\begin{lemma}
\label{lem:extract-vertices}
The tree ${\hat{\cT}}$ can be augmented so that the
above procedure runs in $O(k\log k)$ time and when it ends the only key in $M$ is the root of the extracted tree
\end{lemma}

\begin{proof}
All dictionary operations can be easily implemented in $O(\log k)$ time
whereas the lowest common ancestor can be found in $O(1)$ time after an $O(n)$-time preprocessing (see~\cite{lca}).
This preprocessing requires $O(n)$ space and has to be performed when $\hat{\cT}$ is constructed.
Since we perform $O(k)$ of such operations $O(k \log k)$ is the complexity of our algorithm.\qed
\end{proof}

\subsubsection{Finding the meetings in $\hat{\cT}(S)$}
\ignore{
XXXXXXXXXXXXXXXXXXXXXXXXXXXXXXXXXXXXXXXXXXXXXXXXXXXXXXXXXXXXXXXXx
The intuition behind our method is the following.
Consider a meeting $m$ of two sets (nodes) $Q$ and $R$. What happens with $m$ when we go up the tree $\hat{\cT}(S)$? Consider the parents $Q'$ and $R'$ of $Q$ and $R$, respectively (assume w.l.o.g.\ $\level(Q')\le\level(R')$).
If $Q'=R'$ the meeting $m$ just disappears.
If $Q'\neq R'$ and $\level(Q')=\level(R')$, then $m$ transforms to a meeting $m'(Q',R',\level(Q'))$.
Finally, when $\level(Q')<\level(R')$, then $m$ transforms to a meeting $m'(Q',R,\level(Q'))$.
Hence we can say that any meeting $m$ origins either from a higher meeting $m'$ or from the joining of two sets. Inspired by this observation we will generate meetings in a
top-to-bottom manner.

We scan the nodes of $\hat{\cT}(S)$ in the descending order of their levels.
More precisely, nodes are considered in groups, so that nodes with the same level are analysed in one group.

Now assume we consider a group of nodes $u_1, \ldots, u_m$.
For each such node we are going to find all the meetings it is responsible for.
XXXXXXXXXXXXXXXXXXXXXXXXXXXXXXXXXXXXXXXXXXXXXXXXXXXXXXXXXXXXXXXXXXXXXXXXXXXX}

We generate meetings in a top-down fashion.
We consider the nodes of $\hat{\cT}(S)$ in groups.
Each group corresponds to a single level.
Now assume we consider a group of nodes $u_1, \ldots, u_t$ at some level $\ell$.
Let $v_1,\ldots,v_{t'}$ be the set of children of all nodes $u_i$ in $\hat{\cT}(S)$.
For each node $v_i$, $i=1,\ldots, t'$  we are going to find all the meetings it is responsible for.
Any such meeting ($v_i, x, j)$ is of one of two types:
\begin{enumerate}
 \item $\parent(x)\in\{u_1,\ldots,u_t\}$, possibly $\parent(x)=\parent(v_i)$, or
 \item $\parent(x)\not\in\{u_1,\ldots,u_t\}$, i.e.\ $\level(\parent(x))>\ell$.
\end{enumerate}

\begin{figure}[htbp]
\begin{center}
{\small{
\begin{tikzpicture}[scale=0.68]
\foreach \korzenx in {6cm}
\foreach \korzeny in {2cm}
{

    \draw (-1.5cm,0.3cm) node {$level(u_i)$};
    \draw[dashed, very thin] (-2.5cm, 0cm) -- (12.5cm, 0cm);

    \foreach \x / \y in {9cm/-1cm, 12cm/-1cm}
    {
      \foreach \xx / \yy in {-0.35cm/-0.5cm, 0cm/-0.5cm, 0.35cm/-0.5cm}
      {
        \draw (\x,\y) -- (\x+\xx, \y+\yy);
      }
    }

    \foreach \x / \y in {0cm/0cm, 3cm/0cm, 6cm/0cm, 9cm/1cm, 12cm/1cm}
    {
      \draw (\x, \y) -- (\korzenx, \korzeny);
    }

    \foreach \x / \y in {9cm/1cm, 12cm/1cm}
    {
      \draw (\x, \y) -- (\x, \y-2cm);
    }

    \foreach \x / \y in {0cm/0cm, 3cm/0cm, 6cm/0cm}
    {
      \draw[very thick,dashed] (\x, \y) circle (0.3);
      \fill[black] (\x, \y) circle (0.2);
    }

    \foreach \x / \y in {\korzenx/\korzeny, 9cm/1cm, 12cm/1cm, 9cm/-1cm, 12cm/-1cm}
    {
      \draw[very thick, dashed] (\x, \y) circle (0.3);
      \fill[white] (\x, \y) circle (0.2);
      \draw[very thick] (\x, \y) circle (0.2);
    }

    \foreach \x / \y / \ddy in {0cm/0cm/-0.25cm, 3cm/0cm/0cm, 6cm/0cm/-0.5cm}
    {
      \foreach \dx / \dy / \dolna in {-1cm/-1.2cm/-0.5cm, 0cm/-1cm/0cm, 1cm/-1.4cm/0.5cm}
      {
        \draw (\x,\y) -- (\x+\dx,\y+\dy+\ddy);

        \draw (\x+\dx,\y+\dy+\ddy) -- (\x+\dx,\y+\dy-0.5cm+\ddy);
        \draw (\x+\dx,\y+\dy+\ddy-1.3cm) -- (\x+\dx,\y+\dy-2.1cm+\ddy+\dolna);

        \fill[black] (\x+\dx,\y+\dy+\ddy-0.7cm) circle (0.02);
        \fill[black] (\x+\dx,\y+\dy+\ddy-0.9cm) circle (0.02);
        \fill[black] (\x+\dx,\y+\dy+\ddy-1.1cm) circle (0.02);

        \fill[gray] (\x+\dx,\y+\dy+\ddy) circle (0.2);

        \foreach \xx / \yy in {-0.35cm/-0.5cm, 0cm/-0.5cm, 0.35cm/-0.5cm}
        {
          \draw (\x+\dx,\y+\dy+\ddy-2.1cm+\dolna) -- (\x+\dx+\xx, \y+\dy+\ddy+\yy-2.1cm+\dolna);
        }

        \draw[very thick, dashed] (\x+\dx, \y+\dy+\ddy-2.1cm+\dolna) circle (0.3);
        \fill[white] (\x+\dx, \y+\dy+\ddy-2.1cm+\dolna) circle (0.2);
        \draw[very thick] (\x+\dx, \y+\dy+\ddy-2.1cm+\dolna) circle (0.2);
      }
    }

}
\end{tikzpicture}
}}
\caption{\label{rysunek1}
Extracting meetings. The figure contains a part of tree $\hat{\cT}$.
  Nodes corresponding to the nodes of $\hat{\cT}(S)$ are surrounded by dashed circles.
  The currently processed group of nodes ($u_i$, $i=1,\ldots, k$) are
  filled with black. Nodes from the set $L$ are filled with gray.
  The nodes below the gray nodes are the the nodes $v_j$, i.e.\ the children of nodes $u_i$ in $\hat{\cT}(S)$.}
\end{center}
\end{figure}
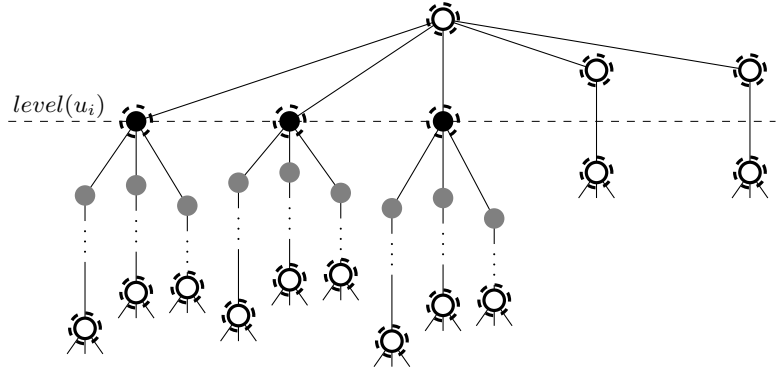

The meetings of the first kind are generated as follows.
Consider the following set of nodes of $\hat{\cT}$ (drawn as grey disks in Figure~\ref{rysunek1}).
\begin{eqnarray}
L = \{x\ :\ \textnormal{$x$ is the first node on the path in $\hat{\cT}$ from $\origin(u_i)$ to $\origin(v_j)$,} \nonumber\\
\textnormal{for some $i=1,\ldots,t$, $j=1,\ldots,t'$}\} \nonumber
\end{eqnarray}
We mark all the nodes of $L$.
Next, we identify all pairs of nodes of $L$ that know each other.
By Lemma~\ref{lem:bounded-responsibility} there are at most $\lambda^{3+\eta}t'=O(t')$  such pairs
and these pairs can be easily found by scanning, for each $x\in L$, all the meetings $x$ is responsible for and such that the node $x$ meets is in $L$.
In this way we identify all pairs of children $(v_i, v_j)$ such that $v_i$ knows $v_j$, namely if $x, y\in L$ and $x$ knows $y$ in $\hat{\cT}$, then $x\cap S$ knows $y\cap S$ in $\hat{\cT}(S)$.
Then, if $v_i$ knows $v_j$, the level of their meeting can be found in $O(\tau\log\lambda\log\log n)$ time using operation $\meet(\origin(v_i), \origin(v_j))$ from Lemma~\ref{lem:binary-jumping}.
Hence, finding the meetings of the first type takes $O(\lambda^{3+\eta}\log\lambda\ \tau t'\log\log n)$ time for one group of nodes, and $O(\lambda^{3+\eta}\log\lambda\ \tau k\log\log n)$ time in total.

Finding the meetings of the second type is easier.
Consider any second type meeting $(v_i, w, l)$. Let $u_j$ be the parent of $v_i$.
Then there is a meeting $(u_j, w, \level(u_j))$ stored in $u_j$.
Hence it suffices to consider, for each $u_j$ all its meetings at level $\level(u_j)$.
For every such meeting $(u_j, w, \level(u_j))$, and for every child $v_i$ of $u_j$ we can apply $\meet(\origin(v_i), \origin(w))$
from Lemma~\ref{lem:binary-jumping} to find the meeting of $v_i$ and $w$.
For the time complexity, note that by Property~\ref{item:acquaintance-limit} of the partition tree, a node $u_j$ meets $\lambda^{3+\eta}=O(1)$ nodes at level $\level(u_j)$. Since we can store the lists of meetings sorted by levels, we can extract all those meetings in $O(\lambda^{3+\eta})$ time. For each meeting we iterate over the children of $u_j$ (Property~\ref{item:childlimit} of the partition tree) and apply Lemma~\ref{lem:binary-jumping}.
This results in $O(\lambda^{3+\eta}\log\lambda\ \tau \log\log n)$ time per a child, hence $O(\lambda^{3+\eta}\log\lambda\ \tau k\log\log n)$ time in total.

After extracting all the meetings, we sort them by levels in $O(k\log k)$ time.

\ignore{
Now as all adjacency edges are extracted, we sort an adjacency list in each node according to the level of those edges.
}
We can claim now the following theorem.

\begin{theorem}
\label{thm:subtree}
For a given set $S \subseteq V$ ($|S|=k$) we can extract the $S$-subtree
of the compressed tree $\hat{\cT}$ in
time $O(\lambda^{3+\eta}\log\lambda\ \tau  k (\log k + \log\log n))=O(k (\log k + \log\log n))$.
\end{theorem}

\section{Pseudospanner construction and applications in approximation}
\label{sec:spanner}

In this section we use the subtree extraction procedure described in the previous section, to construct for any set $S \subseteq V$, a graph that is essentially a small constant stretch spanner for $S$. We then use it to give fast approximations algorithms for several problems.

\subsection{Pseudospanner construction}

\begin{definition}
Let $G=(V,E_G)$ be an undirected connected graph with a weight function $w_G:E_G\rightarrow \mathbb{R}_+$. A graph $H=(V,E_H)$, $E_H \subseteq E_G$ with a weight function $w_H:E_H\rightarrow \mathbb{R}_+$  is an \emph{$f$-pseudospanner} for $G$ if for every pair of vertices $u,v \in V$ we have $d_G(u,v) \le d_H(u,v) \le f \cdot d_G(u,v)$, where $d_G$ and $d_H$ are shortest path metrics induced by $w_G$ and $w_H$.
The number $f$ in this definition is called the \emph{stretch} of the pseudospanner. A pseudospanner for a metric space is simply a pseudospanner for the complete weighted graph induced by the metric space.
\end{definition}

\begin{remark}
Note the subtle difference between the above definition and the classical spanner definition. A pseudospanner $H$ is a subgraph of $G$ in terms of vertex sets and edge sets but it does not inherit the weight function $w_G$. We cannot construct spanners in the usual sense without maintaining the entire distance matrix, which would require prohibitive quadratic space. However, pseudospanners constructed below become classical spanners when provided the original weight function.

Also note, that it immediately follows from the definition of a pseudospanner that for all $uv \in E_H$ we have $w_G(u,v) \le w_H(u,v)$.
\end{remark}

In the remainder of this section we let $(V,d)$ be a metric space of size $n$, where $d$ is doubling with doubling constant $\lambda$. We also use $\hat{\cT}$ to denote the hierarchical tree data structure corresponding to $(V,d)$, and
$\eta$ and $\tau$ denote the parameters of $\hat{\cT}$. For any $S \subset V$, we use $\hat{\cT}(S)$ to denote the subtree of $\hat{\cT}$ corresponding to $S$, as described in the previous section.
Finally, we define a constant $C(\eta,\tau) = \left( 1 + \left(\frac{\tau}{\tau-1}\right)^2 2^{3-\eta}\right) \tau r_j$.

\begin{theorem}
\label{thm:spanner}
Given $\hat{\cT}$ and set $S \subseteq V$, where $|S|=k$, one can construct a $C(\eta,\tau)$-pseudospanner for $S$ in time $O(k(\log k + \log\log n))$. This spanner has size $O(k)$.
\end{theorem}

The proof is in the appendix.

\begin{remark}
Similarly to Property~\ref{item:remstretchlimit} of the partition tree, we can argue that the above theorem gives a $(1+\varepsilon)$-pseudospanner for any $\varepsilon>0$.
Here, we need to take $\tau=1+\frac{\varepsilon}{3}$ and $\eta=O(\frac{1}{\varepsilon^3})$.
\end{remark}

\begin{remark}
It is of course possible to store the whole distance matrix of $V$ and construct a spanner for any given subspace $S$ using standard algorithms. However, this approach has a prohibitive $\Theta(n^2)$ space complexity.
\end{remark}

\subsection{Applications in Approximation}
\label{sec:spanner-app}

Results of the previous subsection immediately give several interesting approximation algorithms.
In all the corollaries below we assume the tree $\hat{\cT}$ is already constructed.

\begin{corollary}[Steiner Forest]
\label{cor:spanner-steiner-forest}
Given a set of points $S \subseteq V$, $|S|=k$, together with a set of requirements $R$ consisting of pairs of elements of $S$, a Steiner forest with total edge-length at most $2 C(\eta,\tau)$OPT=$(2+\varepsilon)$OPT, for any $\varepsilon>0$ can be constructed in time $O(k(\log^2 k + \log\log n))$.
\end{corollary}
\begin{proof}
We use the $O(m\log^2 n)$ algorithm of Cole et al.\ \cite{fast_steiner} (where $m$ is the number of edges) on the pseudospanner guaranteed by Theorem~\ref{thm:spanner}. This algorithm can give a guarantee $2+\epsilon$ for an arbitrarily small $\varepsilon$.\qed
\end{proof}

Similarly by using the MST approximation for TSP we get
\begin{corollary}[TSP]
\label{cor:spanner-tsp}
Given a set of points $S \subseteq V$, $|S|=k$, a Hamiltonian cycle for $S$ of total length at most $2 C(\eta,\tau)$OPT=$(2+\varepsilon)$OPT for any $\varepsilon>0$ can be constructed in time $O(k(\log k + \log\log n))$.
\end{corollary}

\ignore{
In (uncapacitated) facility location problem we have a set $F$ of facilities and
a set $C$ of cities. For every facility $i\in F$ there is specified an {\em opening cost}, i.e.\ a nonnegative number $f_i$.
Furthermore, for each city-facility pair there is given a {\em connection  cost}.
Here we consider the metric version of the problem, where $F$ and $C$ are subsets of a metric space $(X,d)$ and connection costs are given by function $d$.
The objective is to open a subset of facilities in $F$, and connect each city to an open facility so that the total cost is minimized.
}

Currently, the best approximation algorithm for the facility location problem is
the $1.52$-approximation of Mahdian, Ye and Zhang~\cite{mahdian:1.52}.
A fast implementation using Thorup's ideas~\cite{thorup:quick-and-good} runs in
deterministic $O(m\log m)$ time, where $m=|F|\cdot |C|$, and
if the input is given as a weighted graph of $n$ vertices and $m$ edges, in $\tilde{O}(n+m)$ time, with high probability (i.e.\ with probability $\ge 1-1/n^{\omega(1)}$).
In an earlier work, Thorup~\cite{thorup:quick} considers also the $k$-center and $k$-median problems in the graph model. When the input is given as a weighted graph of $n$ vertices and $m$ edges, his algorithms run in $\tilde{O}(n+m)$ time, w.h.p.\ and have approximation guarantees of $2$ for the $k$-center problem and $12+o(1)$ for the $k$-median problem.
By using this latter algorithm with our fast spanner extraction we get the following corollary.

\begin{corollary}[Facility Location with restricted facilities]
\label{cor:spanner-facility-location}
Given two sets of points $C \subseteq V$ (cities) and $F \subseteq V$ (facilities) together with  opening cost $f_i$ for each facility $i\in F$, for any $\varepsilon>0$, a $(1.52+\varepsilon)$-approximate solution to the facility location problem can be constructed in time $O((|C|+|F|)(\log^{O(1)}(|C|+|F|)+\log\log|V|))$, w.h.p.
\end{corollary}

The application of our results to the variant of {\sc Facility Location} with unrestricted facilities is not so immediate.
We were able to obtain the following.

\begin{theorem}[Facility Location with unrestricted facilities]
\label{thm:only-cities-facility-location}
Assume that for each point of $n$-point $V$ there is assigned an opening cost $f(x)$.
Given a set of $k$ points $C \subseteq V$, for any $\varepsilon>0$, a $(3.04+\varepsilon)$-approximate solution to the facility location problem with cities' set $C$ and facilities' set $V$ can be constructed in time $O(k\log k (\log^{O(1)}k+\log\log n))$, w.h.p.
\end{theorem}

The above result is described in Appendix~\ref{sec:unres-fl}.
Our approach there is a reduction to the variant with restricted facilities.
The general, rough idea is the following: during the preprocessing phase, for every point $x\in V$ we compute a small set $F(x)$ of facilities that seem a good choice for $x$, and when processing a query for a set of cities $C$, we just apply Corollary~\ref{cor:spanner-facility-location} to cities' set $C$ and facilities' set $\bigcup_{c\in C}F(c)$.

\begin{corollary}[$k$-center and $k$-median]
\label{cor:spanner-k-center}
Given a set of points $C \subseteq V$ and a number $r \in \mathbb{N}$, for any $\varepsilon>0$, one can construct:
\begin{enumerate}
 \item[(i)] a $(2+\varepsilon)$-approximate solution to the $r$-center problem, or
 \item[(ii)] a $(12+\varepsilon)$-approximate solution to the $r$-median problem
\end{enumerate}
in time $O(|C|(\log|C|+\log\log|V|))$, w.h.p.
\end{corollary}

\section{Dynamic Minimum Spanning Tree and Steiner Tree}
\label{sec:dynamic}

\newcommand{\OPT}{\ensuremath{\textrm{OPT}}}

In this section we give one last application of our hierarchical data structure. It has a different flavour from the other applications presented in this paper since it is not based on constructing a spanner, but uses the data structure directly.
We solve the Dynamic Minimum Spanning Tree / Steiner Tree (DMST/DST) problem,
where we need to maintain a spanning/Steiner tree of a subspace $X \subseteq V$ throughout a sequence of vertex additions and removals to/from $X$.

The quality of our algorithm is measured by the total cost of the tree produced relative to the optimum tree, and time required to add/delete vertices.
Let $|V|=n$, $|X|=k$.
Our goal is to give an algorithm that maintains a constant factor approximation of the optimum tree, while updates are polylogarithmic in $k$, and do not depend (or depend only slightly) on $n$.
It is clear that it is enough to find such an algorithm for DMST. Due to space limitations, in this section we only formulate the results.
Precise proofs are gathered in Appendix \ref{app:dynamic}.

  \begin{theorem}\label{thm:dynamic-mst}
  Given the compressed tree $\hat{\cT}(V)$, we can maintain an $O(1)$-approximate Minimum Spanning Tree for a subset $X$ subject to
  insertions and deletions of vertices.
  The insert operation works in $O(\log^5k + \log\log n)$ time and
  the delete operation works in $O(\log^5k)$ time, $k=|X|$.
  Both times are expected and amortized.
  \end{theorem}

\bibliographystyle{plain}
{\small
\bibliography{jiatree}
}

\appendix
\newpage
\section{Related Work}
\label{app:related_work}
In the next few paragraphs we review different approaches to this problem, state the differences and try to point out the advantage of the results presented here.

\paragraph{Universal Algorithms}
In the case of {\sc Steiner Tree} and {\sc TSP} results pointing in the direction studied here have been already obtained. In the so called, universal approximation algorithms introduced by Jia {\it et. al}~\cite{jia}, for each element of the request we need to fix an universal solution in advance.
More precisely, in the case of {\sc Steiner Tree} problem for each $v\in V$ we fix a path $\pi_v$, and a solution to $S$ is given as $\{\pi_v:v\in S\}$. Using universal algorithms we need very small space to remember the precomputed solution and we are usually able to answer queries efficiently, but the corresponding approximation ratios are relatively weak, i.e, for {\sc Steiner Tree} the approximation ratio is $O(\log^4 n/\log \log n)$. Moreover, there is no direct way of answering queries in $\tilde{O}(k)$ time, and in order to achieve this bound one needs to use similar techniques as we use in Section~\ref{sec:fast-subtree-extraction}. 
In our model we loosen the assumption that the solution itself has to be precomputed beforehand, but the data output of the preprocessing is of roughly the same size (up to polylogarithmic factors). Also, we allow the algorithm slightly more time for answering the queries and, as a result are able to improve  the approximation ratio substantially --- from polylogarithmic to a constant.

\paragraph{Spanners and Distance Oracles}
The question whether the graph can be approximately represented using less space than its size was previously captured by the notion of spanners and approximate distance oracles. Both of these data structures represent the distances in the graphs up to a given multiplicative factor $f$. The difference is that the spanner needs to be a subgraph of the input graph hence distances between vertices are to be computed by ourselves, whereas the distance oracle can be an arbitrary data structure that can compute the distances when needed. However, both are limited in size. For general graphs $(2t-1)$-spanners (i.e., the approximation factor is $f=2t-1$) are of size $O(n^{1+1/t})$ and can be constructed in randomized linear time as shown by Baswana and Sen~\cite{BS03}. On the other hand, Thorup and Zwick~\cite{tz05} have shown that the $(2t-1)$-approximate oracles of size $O(tn^{1+1/t})$, can be constructed in $O(tmn^{1+1/t})$ time, and are able to answer distance queries in $O(t)$ time. It seems that there is no direct way to obtain, based on these results, an algorithm that could answer our type of queries faster then $O(k^2)$.

The construction of spanners can be improved in the case of doubling metric. The papers~\cite{spanner_d1,spanner_d2} give a construction of $(1+\epsilon)$-spanners that have linear size in the case when $\epsilon$ and the doubling dimension of the metric are constant. Moreover, Har-Peled and Mendel~\cite{spanner_d1} give $O(n \log n)$ time construction of such spanners.
A hierarchical structure similar to that of~\cite{kr-lee} and the one we use in this paper was also used by Roditty~\cite{roditty} to maintain a dynamic spanner of a doubling metric, with a $O(\log n)$ update time. However, all these approaches assume the existence of a distance oracle.
When storing the whole distance matrix, these results, combined with known approximation algorithms in the classical setting~\cite{mahdian:1.52,thorup:quick-and-good,thorup:quick,fast_steiner}, imply a data-structure that can answer {\sc Steiner Tree}, {\sc Facility Location} with restricted facilities and {\sc $k$-Median} queries in $O(k\log k)$ time. However, it does not seem to be easy to use this approach to solve the variant of {\sc Facility Location} with unrestricted facilities. To sum  up, spanners seem to be a good solution in our model in the case when a $O(n^2)$ space is available for the data structure. The key advantage of our solution is the low space requirement. On the other hand, storing the spanner requires nearly linear space, but then we need $\tilde{O}(n)$ time to answer each query. The distance matrix is unavailable and we will need to process the whole spanner to respond to a query on a given set of vertices.  

\paragraph{Sublinear Approximation Algorithms}
Another way of looking at the problem is the attempt to devise sublinear algorithm that would be able to solve approximation problems for a given metric. This study was started by Indyk~\cite{indyk99} who gave constant approximation ratio $O(n)$-time algorithms for: {\sc Furthest Pair}, {\sc $k$-Median} (for constant $k$), {\sc Minimum Routing Cost Spanning Tree}, {\sc Multiple Sequence Alignment}, {\sc Maximum Traveling Salesman Problem}, {\sc Maximum Spanning Tree} and {\sc Average Distance}. Later on B\u{a}doiu {\it et. al}~\cite{indyk05} gave an $O(n\log n)$ time algorithm for computing the cost of the uniform-cost metric {\sc Facility Location} problem. These algorithms work much faster that the $O(n^2)$-size metric description. However, the paper  contains many negative conclusions as well. The authors show that for the following problems $O(n)$-time constant approximation algorithms do not exists: general metric {\sc Facility Location}, {\sc Minimum-Cost Matching} and {\sc $k$-Median} for $k=n/2$. In contrary, our results show that if we allow the algorithm to preprocess partial, usually fixed, data we can answer queries in sublinear time afterwards.

\paragraph{Dynamic Spanning Trees}
The study of online and dynamic Steiner tree was started in the paper of \cite{onlie-steiner-tree}. However, the model considered there was not taking the computation time into account, but only minimized the number of edges changed in the Steiner tree. More recently the Steiner tree problem was studied in a setting more related to ours~\cite{dst2,dst3,dst5,dst4}. The first three of these paper study the approximation ratio possible to achieve when the algorithm is given an optimal solution together with the change of the data. The efficiency issue is only raised in~\cite{dst4}, but the presented algorithm in the worst case can take the same as computing the solution from scratch. The problem most related to our results is the dynamic minimum spanning tree (MST) problem. The study of this problem was finished by showing deterministic algorithm supporting edge updates in polylogarithmic time in~\cite{hlt:dynamic-mst}.  The dynamic Steiner tree problem is a direct generalization of the dynamic MST problem, and we were able to show similar time bounds. However, there are important differences between the two problems that one needs to keep in mind. In the case of MST, by definition, the set of terminals remains unchanged, whereas in the dynamic Steiner tree we can change it. On the other hand we cannot hope to get polylogarithmic update times if we allow to change the edge weights, because this would require to maintain dynamic distances in the graph. The dynamic distance problem seems to require polynomial time for updates~\cite{dynamic-shortest-paths}. 

\section{Partition tree --- precise definitions and proofs}\label{app:tree}

To start with, let us recall partition and partition scheme definitions.

\begin{definition}[Jia et al \cite{jia}, Definition 1] A $(r, \sigma, I)$-partition is
a partition of $V$ into disjoint subsets ${S_i}$ such that $diam\ S_i \leq r\sigma$ for all $i$
and for all $v \in V$, the ball $B_r(v)$ intersects at most $I$ sets in the partition.

A $(\sigma, I)$ partition scheme is an algorithm that produces $(r, \sigma, I)$-partition
for arbitrary $r \in \cR, r > 0$.
\end{definition}

\begin{lemma}[similar to Jia et al \cite{jia}, Lemma 2]\label{lem:doubpart}
Let $\eta \geq 0$ be a nonnegative integer.
For $V$ being a doubling metric space with doubling constant $\lambda$, there exists $(2^{-\eta}, \lambda^{3+\eta})$ partition scheme that works in polynomial time.
Moreover, for every $r$ the generated partition $\cS_r$ has the following property: for every $S \in \cS_r$ there exists $\opleader(S) \in S$ such that
$S \subset B_{2^{-\eta-1}r}(\opleader(S))$.
\end{lemma}

\begin{proof}
Take arbitrary $r$. Start with $V_0 = V$. At step $i$ for $i=0,1,\ldots$ take any $v_i \in V_i$ and take $S_i = B_{2^{-\eta-1}r}(v_i) \cap V_i$.
Set $V_{i+1} = V_i \setminus S_i$ and proceed to next step. Obviously, $S_i \subset B_{2^{-\eta-1}r}(v_i)$, so $diam\ S_i < 2^{-\eta}r$ and
we set $\opleader(S_i) = v_i$.

Take any $v \in V$ and consider all sets $S_i$ crossed by ball $B_r(v)$. Every such set is contained in $B_{(1 + 2^{-\eta})r}(v) \subset B_{2r}(v)$, which can be covered
by at most $\lambda^{3+\eta}$ balls of radius $2^{-\eta-2}r$. But for every $i \neq j$, $d(v_i, v_j) > 2^{-\eta-1}r$, so every leader of set crossed by $B_r(v)$
must be in a different ball. Therefore there are at most $\lambda^{3+\eta}$ sets crossed.\qed
\end{proof}



Let us define the space partition tree $\cT$.

\begin{algorithm}\label{alg:firsttree}
Assume we have doubling metric space $(V, d)$ and $(2^{-\eta}, \lambda^{3+\eta})$ partition scheme form Lemma \ref{lem:doubpart}.
Let us assume $\eta \geq 2$ and let $\tau$ be a real constant satisfying:
\begin{itemize}
\item $2\frac{\tau 2^{-\eta}}{\tau - 1} \leq 1$, i.e, $\tau \geq \frac{1}{2^{\eta-1} - 1} + 1$.
\item $\tau \leq 2^\eta$.
\end{itemize}
Then construct space partition tree $\cT$ as follows:
\begin{enumerate}
\item Start with partition $\cS_0 = \{\{v\} : v \in V\}$, and $r_0 < \min \{d(u, v): u, v \in V, u \neq v\}$.
For every $\{v\} \in \cS_0$ let $\opleader(\{v\}) = v$. Let $\cS_0' = \cS_0$.
\item Let $j := 0$.
\item While $\cS_j$ has more than one element do:
  \begin{enumerate}
  \item Fix $r_{j+1} := \tau r_j = \tau^j r_0$.
  \item Let $\cS_{j+1}'$ be a partition of the set $L_j = \{\opleader(S): S \in \cS_j\}$ generated by given partition scheme for $r=2r_{j+1}$.
  \item Let $\cS_{j+1} := \{\bigcup \{S : \opleader(S) \in S'\}: S' \in \cS_{j+1}'\}$.
  \item Set $\opleader(\bigcup \{S : \opleader(S) \in S'\}) = \opleader(S')$ for any $S' \in \cS_{j+1}'$.
  \item $j := j + 1$.
  \end{enumerate}
\end{enumerate}
\end{algorithm}

Note that for every $j$, $\cS_j$ is a partition of $V$.
We will denote by $\opleader_j(v)$ the leader of set $S \in \cS_j$ that $v \in S$.

\begin{definition}
We will say that $S^* \in \cS_{j+1}$ is a parent of $S \in \cS_j$ if $\opleader(S) \in S^*$ (equally $S \subset S^*$).
This allows us to consider sets $\cS_j$ generated by Algorithm \ref{alg:firsttree} as nodes of a tree $\cT$ with root being the set $V$.
\end{definition}

\begin{lemma}\label{lem:leaderdist}
For every $j$ and for every $v \in S$ the following holds:
$$d(v, \opleader_j(v)) < \frac{\tau 2^{-\eta}}{\tau - 1} r_j.$$
\end{lemma}

\begin{proof}
Note that
$$d(v, \opleader_j(v)) \leq \sum_{i=1}^{j} d(\opleader_i(v), \opleader_{i-1}(v))$$
We use bound from Lemma \ref{lem:doubpart}:
$$\sum_{i=1}^{j} d(\opleader_i(v), \opleader_{i-1}(v)) \leq \sum_{i=1}^{j} 2^{-\eta-1} \cdot 2\tau^i r_0 = 2^{-\eta} \tau \frac{\tau^j - 1}{\tau - 1} r_0 < \frac{\tau 2^{-\eta}}{\tau-1}r_j.$$\qed
\end{proof}

\begin{lemma}\label{lem:doubknow}
For every $j$, for every $S \in \cS_j$, the union of balls $\bigcup \{B_{r_j}(v): v \in S\}$
crosses at most $\lambda^{3 + \eta}$ sets from the partition $\cS_j$.
\end{lemma}

\begin{proof}
For $j = 0$ this is obvious, since $r_0$ is smaller than any $d(u,v)$ for $u \neq v$. Let us assume $j > 0$.

Let $v \in S \in \cS_j$, $v^* \in S^* \in \cS_j$, $S \neq S^*$ and $d(v, v^*) < r_j$.
Then, using Lemma \ref{lem:leaderdist},
$$d(\opleader_j(v), \opleader_j(v^*)) \leq d(\opleader_j(v), v) + d(v, v^*) + d(v^*, \opleader_j(v^*)) <$$
$$< r_j \Big(1 + 2\frac{\tau 2^{-\eta}}{\tau - 1}r_j\Big) < 2r_j$$
Since, by partition properties, $B_{2r_j}(\opleader_j(v))$ crosses at most $C$ sets
from $\cS_j'$ and $\opleader_j(v^*) \in B_{2r_j}(\opleader_j(v))$, this finishes the proof.\qed
\end{proof}

\begin{definition}
We say that a set $S \in \cS_j$ {\em{knows}} a set $S' \in \cS_j$ if
$\bigcup \{B_{r_j}(v): v \in S\} \cap S' \neq \emptyset$.
We say that $v \in V$ knows $S' \in \cS_j$ if $v \in S \in \cS_j$ and $S$ knows $S'$ or $S = S'$.
\end{definition}

Note that Lemma \ref{lem:doubknow} implies the following:
\begin{corollary}\label{cor:acquaintance-limit}
A set (and therefore a node too) at a fixed level $j$ has at most $\lambda^{3+\eta}$ acquaintances.
\end{corollary}

\begin{lemma}\label{lem:parentknows}
Let $S \in \cS_j$ be a child of $S^* \in \cS_{j+1}$ and let $S$ know $S'\in \cS_j$.
Then either $S' \subset S^*$ or $S^*$ knows the parent of $S'$.
\end{lemma}

\begin{proof}
Assume that $S'$ is not a child (subset) of $S^*$ and let $S^{**} \in \cS_{j+1}$ be the parent of $S'$. Since $S$ knows $S'$,
there exist $v \in S$, $v' \in S'$ satisfying $d(v, v') < r_j$. But $r_j < r_{j+1}$ and $v \in S^*$ and $v' \in S^{**}$.
\qed
\end{proof}

\begin{lemma}\label{lem:childlimit}
Set $S^* \in \cS_j$ has at most $\lambda^{3+\eta}$ children in the tree $\cT$.
\end{lemma}

\begin{proof}
By construction of level $j$, let $S \in \cS_{j-1}$ be such a set that $\opleader(S) = \opleader(S^*)$
(in construction step we divided sets of leaders $L_{j-1}$ into partition $\cS_j'$). Let $S' \in \cS_{j-1}$
be another child of $S^*$. Then, by construction and assumption that $\tau \leq 2^\eta$:
$$d(\opleader(S'), \opleader(S)) < 2r_j \cdot 2^{-\eta-1} = 2^{-\eta}r_j \leq r_{j-1}.$$

However, by Lemma \ref{lem:doubknow}, $B_{r_{j-1}}(\opleader(S))$ crosses at most $\lambda^{3+\eta}$ sets at level $j-1$.
That finishes the proof.\qed
\end{proof}

\begin{lemma}\label{lem:stretchlimit}
Let $v, v^* \in V$ be different points such that $v \in S_1 \in \cS_j$, $v \in S_2 \in \cS_{j+1}$ and $v^* \in S_1^* \in \cS_j$, $v^* \in S_2^* \in \cS_{j+1}$ and
$S_2$ knows $S_2^*$ but $S_1$ does not know $S_1^*$. Then
$$r_j \leq d(v, v^*) < \Big( 1 + \frac{4\tau 2^{-\eta}}{\tau - 1}\Big) \tau r_j.$$

For $\tau = 2$ and $\eta = 2$ this implies $r_j \leq d(v, v^*) \leq 6r_j$.
\end{lemma}

\begin{proof}
Since $S_1$ and $S_1^*$ do not know each other, $v$ and $v^*$ are in distance at least $r_j$. Since $S_2$ knows $S_2^*$, there
exist $u \in S_2$ and $u^* \in S_2^*$ such that $d(u, u^*) < r_{j+1}$. Therefore
$$d(v, v^*) \leq$$
$$\leq d(v, \opleader(S_2)) + d(\opleader(S_2), u) + d(u, u^*) + $$
$$+d(\opleader(S_2^*), u^*) + d(\opleader(S_2^*), v^*) < $$
$$< 4 \cdot \frac{\tau 2^{-\eta}}{\tau - 1} r_{j+1} + r_{j+1} = \Big( 1 + \frac{4\tau 2^{-\eta}}{\tau - 1}\Big) \tau r_j.$$
\qed
\end{proof}

\begin{remark}\label{rem:stretchlimit}
Imagine we want in Lemma \ref{lem:stretchlimit} to obtain bound $r_j \leq d(v, v^*) \leq (1+\varepsilon) r_j$ for some
small $1 > \varepsilon > 0$. Take $\tau = 1 + \frac{\varepsilon}{3}$.
We want here the following:
$\frac{4\tau 2^{-\eta}}{\tau - 1} < \frac{\varepsilon}{3},$
i.e.,
$2^{-\eta} < \frac{\varepsilon^2}{12(1+\varepsilon)} < \frac{\varepsilon^2}{24}.$
Then we have
$$d(v, v^*) < \Big( 1 + \frac{4\tau 2^{-\eta}}{\tau - 1}\Big) \tau r_j < \Big(1+\frac{\varepsilon}{3}\Big)^2 r_j < (1+\varepsilon)r_j.$$
Note, that to obtain this we need $2^\eta = O(\frac{1}{\varepsilon^2})$.
Note, that conditions in Algorithm \ref{alg:firsttree} for $\eta$ and $\tau$ are much weaker than we assumed here.
\end{remark}

\section{Implementation of the $\meet$ and $\jump$ operations}
\label{app:meet}

In this section we provide realizations of $\meet$ and $\jump$ operations that work fast, i.e., roughly in $O(\log \log n)$ time.

Let us now recall the semantics of the $\meet$ operation, which was used in the fast subtree extraction in Section~\ref{sec:fast-subtree-extraction}.
For nodes $u$ and $v$, by $u(j)$ and $v(j)$ we denote the ancestor of $u$ (resp. $v$) in the tree at level $j$.
The $\meet(v, u)$ operation returns the lowest level $j$ such that $u(j)$ and $v(j)$ knows each other.
This operation can be performed in $O(\lambda^{\eta + 3} \log \log n)$ time.

Operation $\jump$ is used by the dynamic algorithms in Section~\ref{sec:dynamic}, and its semantics is as follows.
In the compressed tree, for each set $S$ we store
a list of all meetings of $S$, sorted by level. The $\jump(v, i)$, given node $v$ and level $i$ outputs
the set $S$ and a meeting $(S, S', j)$ such that $v \in S$ and $j$ is the lowest possible level such that $i \leq j$.
Informally speaking, it looks for the first meeting of a set containing $v$ such that its level is at least $i$.
The $\jump$ operation works in $O(\log \log n + \log \log \log \metricstretch)$.
If we require that there is some meeting at level $i$ somewhere, maybe distant from $v$, in the tree,
the time reduces to $O((\log \eta + \log \log \lambda)\log \log n)$.

\subsection{Path partition}

In order to implement the $\jump$ and $\meet$ operations efficiently we need to store additional information concerning
the structure of $\hat{\cT}$, namely a path partition. The following lemma defines the notion.

\begin{lemma}
\label{lem:tree-partition}
The set of edges of the tree $\hat{\cT}$ can be partitioned into a set of paths $\cP=\{P_1,\ldots,P_m\}$
such that each path starts at some node of $\hat{\cT}$ and goes down the tree only
and for each node $v$ of the tree $\hat{\cT}$ the path from $v$ to the root
contains edges from at most $\lceil \log_2n \rceil$ paths of the path decomposition $\cP$.
Moreover $\cP$ can be found $O(n)$ time.
\end{lemma}

\begin{proof}
We use a concept similar to the one used by Sleator and Tarjan in~\cite{tarjan}.
We start from the root and each edge incident to the root is a beginning of a new path.
We then proceed to decompose each subtree of the root recursively.
When considering a subtree rooted at a node $v$ we lengthen the path
going down from the parent of $v$ by one edge going to the subtree containing
the largest number of nodes (breaking ties arbitrarily).
Each of the remaining edges leaving $v$ starts a new path.

It is easy to see that each path goes down the tree only. Now consider a node $v$.
When we go up from $v$ to the root, every time we reach an end of some path from
$\cP$, the size of the subtree rooted at the node we move
into doubles. This ends the proof since there are at most $2n-1$ vertices.\qed
\end{proof}

We now describe additional information related to the path decomposition that we need to store.
Each node $v$ of $\hat{\cT}$ maintains a set $\paths$, where $(i,level) \in \paths(v)$ if
the path from $v$ to the root contains at least one edge of the path $P_i$, and the lowest such
edge has its bottom endpoint at level $level$. In other words, $P_i$ enters the path from
$v$ to the root at level $level$.
We use two different representations of the set $\paths$ simultaneously. One is a dictionary
implemented as a hash table, and the other is an array sorted by $level$.
Because of the properties of the path decomposition $\cP$ from Lemma~\ref{lem:tree-partition}
for each node $v$ we have $|\paths(v)| \le \lceil \log_2(n) \rceil$.

Let $P_i \in \cP$ be a path with vertices $\{v_1,\ldots,v_t\}$ (given in order of increasing level).
We define $\interior(P_i)$ to be the set $\{v_1,\ldots,v_{t-1}\}$, i.e.\ we exclude the top vertex of $P_i$.
We also define $\toplevel(P_i)$ to be the level of $v_{t-1}$, i.e.\ the highest level among interior nodes
of $P_i$.

\subsection{The $\meet$ operation}

In order to benefit from the path decomposition to implement $\meet$ operation,
we also need to store adjacency information for
paths, similar to the information we store for single nodes.
Let $P_a,P_b \in \cP$ be two paths, such that their interior
nodes know each other at level $j_{ab}$, but not at level $j_{ab}-1$.
Then the triple $(P_a,P_b,j_{ab})$ is called a {\em meeting} of
$P_a$ and $P_b$ at level $j_{ab}$. We also say that $P_a$ and $P_b$ {\em meet}
at level $j_{ab}$), or that they know each other.
This definition is just a generalisation of a similar definition
for pairs of nodes of $\cT$. We may also define a notion of {\em responsibility}
for paths which is analogous to the definition for nodes
and formulate a lemma analogous to Lemma~\ref{lem:know-query}.

\begin{lemma}
\label{lem:know-path}
One can augment the tree $\hat{\cT}$ with additional information of size $O(n\lambda^{3 + \eta})$,
so that for any pair of paths $P_x, P_y \in \cP$ one can decide if $P_x$ and $P_y$ know each other, and if that is the case the level of the meeting is returned.
The whole query takes $O(\eta\log\lambda)$ time.
\end{lemma}

Now, suppose we are given two nodes $u,v \in \hat{\cT}$ and we are to compute $\meet(u,v)$.
The following lemma provides a crucial insight into how this can be done.

\begin{lemma}
\label{lem:path-part}
Let $(i,j) \in \paths(u)$, which means that the path $P_i$ reaches
the path from $u$ to the root at level $j$ and assume
that nodes $u,v$ start to know each other
at level $j_{uv}=\meet(u,v)$, where $j_{uv} \le \toplevel(P_i)$.
Then either $(i,\ell)\in \paths(v)$ for some $\ell$, or
there exists $i'$, such that paths $P_i$ and $P_{i'}$
know each other, $P_i$ is responsible for their meeting, and $(i',\ell) \in \paths(v)$ for some $\ell$.
Moreover, this condition can be checked in $O(\lambda^{\eta + 3})$ time.
\end{lemma}

\begin{proof}
Since $j_{uv} \le \toplevel(P_i)$ we know that at level $\toplevel(P_i)$
paths from $u$ to the root and from $v$ to the root either merged,
or else nodes on those paths at level $\toplevel(P_i)$ know each other.
If those paths merged, than $P_i$ intersects the path from $v$ to
the root, and we know that $(i,*)\in \paths(v)$. This can be checked in
hash table for $\paths(v)$ in $O(1)$ time.

Otherwise as $i'$ we take $P_i$ to be the lowest path $P_{i''} \in \cP$,
such that $(i'',\ell) \in \paths(v)$ for some $\ell$, and $\toplevel(P_i'') \ge \toplevel(P_i)$.
To check if this occurs, we take $S_i$ --- the interior node of $P_i$ with the highest level,
and iterate over all $S_i'$ known by $S_i$ and look for path containing $S_i'$ in the
hashtable for $\paths(v)$. As $S_i$ knows at most $\lambda^{\eta+3}$ sets, the bound follows.\qed
\end{proof}

Now, using Lemma~\ref{lem:path-part} we can do a binary search
over the elements of $\paths(u)$, and find a pair $(i_u,j_u) \in \paths(u)$
such that $\meet(u,v) \le \toplevel(P_{i_u})$ and $\meet(u,v) \ge j_u$.
Namely, we look for the lowest path in $\paths(u)$ that satisfies Lemma~\ref{lem:path-part}.
Similarly, we can find $(i_v,j_v) \in \paths(v)$.
Since paths $P_{i_u}$ and $P_{i_v}$ know each other, we simply
use Lemma~\ref{lem:know-path} to find the exact level $j$ where they meet,and as the result
of $\meet(u,v)$ return $\max(j_u,j_v,j)$.
We need to take the maximum of those values, because paths $P_{i_u}$
and $P_{i_v}$ could possibly meet before they enter the paths from $u$ and $v$ to the root.

\begin{lemma}[Lemma~\ref{lem:binary-jumping} restated]
\label{lem:binary-jumping-app}
The tree $\hat{\cT}$ can be augmented so that the $\meet$ operation can be
performed in $O(\eta\log\lambda \log \log n)$ time.
The augmented $\cT$ tree can be stored in $O(\lambda^{3+\eta}n\log n)$ space and computed in polynomial time.
\end{lemma}

\begin{proof}
Since $|\paths(u)| \le \lceil \log_2n \rceil$ we perform $O(\log \log n)$ steps
of the binary search. During each step we perform $O(\lambda^{\eta+3})$
searches in a hash table, thus we can find the result of $\meet(u,v)$ in $O(\log\log n)$ time.

The space bound follows from Corollary~\ref{lem:num_edges} (the additional $\log n$ factor in the space bound comes from the size of $\paths(x)$ for each node $x$). Now we need only to describe how to obtain running time independent of the stretch of the metric.
In order to compute the $\hat{\cT}$ tree (without augmentation) we can slightly improve our construction algorithm:
instead of going into the next level, one can compute the smallest distance between current sets and jump directly
to the level when some pair of sets merges or begins to know each other.\qed
\end{proof}

\begin{remark}
 We could avoid storing $\paths$ in arrays by maintaining, for each path in $\cP$, links to paths distant by powers
 of two in the direction of the root (i.e.\ at most $\log\log n$ links for each path).

Also, to obtain better space bound, we could use a balanced tree instead of the hash tables
 to keep the first copy of $\paths$. If we use persistent balanced trees, we can get an $O(n\log \log n)$ total space bound.
 However, in that case the search time would be increased to $O((\log\log n)^2)$ for one call to the $\meet$ operation.
\end{remark}

\subsection{The $\jump$ operation}

\begin{lemma}\label{lem:jump}
  The compressed tree $\hat{\cT}$ can be enhanced with additional information of size
  $O(\lambda^{\eta+3}n\log n)$ in such a way that the $\jump(v, i)$ operation can be performed in
  $O(\log \log n + \log \log \log \metricstretch)$ time, where $\metricstretch$
  denotes the stretch of the metric. If we require that there is some meeting
  at level $i$ somewhere in the tree (possibly not involving $v$), the $\jump$ operation
  can be performed in $O((\log \eta + \log \log \lambda)\log\log n)$ time.
\end{lemma}

\begin{proof}
To calculate $\jump(v, i)$, we first look at $\paths(v)$ and binary search lowest path $P \in \paths(v)$ such
that the highest node in $P$ has level greater than $i$. If $P = \{v_1, \ldots, v_t\}$ (given in order of increasing level),
that means that $\level(v_1) \leq i < \level(v_t)$. This step takes $O(\log \log n)$ time.

To finish the $\jump$ operation, we need, among meetings on path $P$, find the lowest
one with the level not smaller than $i$. As levels are numbered from $0$ to $\log \metricstretch$,
this can be done using y-Fast Tree data structure \cite{Willard84,Willard86}.
The y-Fast Tree uses linear space and answers predecessor queries in
$O(\log \log u)$ time, where $u$ is the size of the universe, here $u = \log \metricstretch$.

To erase dependency on $\metricstretch$, note that according to Corollary \ref{lem:num_edges},
where are at most $M := (2n-1)\lambda^{\eta+3}$ meetings in tree $\hat{\cT}$. Therefore,
we can assign to every level $j$, where some meeting occurs, a number $0 \leq n(j) < M$
and for two such levels $j$ and $j'$, $j < j'$ iff $n(j) < n(j')$. The mapping $n(\cdot)$
can be implemented as a hash table, thus calculating $n(j)$ takes $O(1)$ time.
Instead of using y-Fast Trees with level numbers as universe, we use numbers $n(\cdots)$.
This requires $O(\log \log n + \log \log M) = O((\log \eta + \log \log \lambda)\log \log n)$ time,
but we need to have the key $i$ in the hash table, i.e., there needs to be some meeting
at level $i$ somewhere in the tree.\qed
\end{proof}

\section{Omitted Proofs}
\label{app:omitted}


\begin{proof}[of Theorem \ref{thm:spanner}]

Recall that nodes of $\hat{\cT}(S)$ are simply certain subsets of $S$, in particular all single-element subsets of $S$ are nodes of $\hat{\cT}(S)$. Associate with every node $A$ of $\hat{\cT}(S)$, an element $a$ of $A$, which we will call $\opleader(A)$, so that:
\begin{itemize}
\item if $A = \{a\}$ (which means $A$ is a leaf in $\hat{\cT}(S)$), then $\opleader(A) = a$,
\item if $A$ has sons $A_1,\ldots,A_m$ in $\hat{\cT}(S)$, then let $\opleader(A)$ be any of $\opleader(A_i)$, $i=1,\ldots,m$.
\end{itemize}
If two nodes $A,B$ in $\hat{\cT}(S)$ know each other, we will also say that their leaders $\opleader(A)$ and $\opleader(B)$ know each other. Also, if $A$ is the parent of $B$, and $a \neq b$, where $a=\opleader(A)$ and $b=\opleader(B)$, we will say that $a$ is the parent of $b$. We will also say that $a$ \emph{beats} $b$ at level $L$, where $L$ is the level at which $A$ appears as a node --- this is exactly the level where $b$ stops being a leader, and is just an ordinary element of a set where $a$ is a leader.

Now we are ready do define the pseudospanner.
Let $H=(S,E)$, where $E$ contains all edges $uv$, $u \neq v$ such that:
\begin{enumerate}
\item $v$ is the father of $u$, or
\item $u$ and $v$ know each other.
\end{enumerate}
We cannot assign to these edges their real weights, because we do not know them. Instead, we define $w_H(u,v)$ to be an upper bound on $d(u,v)$, which is also a good approximation of $d(u,v)$. In particular:
\begin{enumerate}
\item If $u$ is a son of $v$ and $v$ beats $u$ at level $j$, we put $w_H(u,v) = 2\frac{\tau 2^{-\eta}}{\tau - 1} r_j$.
\item If $u$ and $v$ first meet each other at level $j$, we put $w_H(u,v) = \Big( 1 + \frac{4\tau 2^{-\eta}}{\tau - 1}\Big) \tau r_j$.
\end{enumerate}

We claim that $H$ is a $C(\eta,\tau)$-spanner for $V$ of size $O(n)$.

It easily follows from Lemmas~\ref{lem:leaderdist} and~\ref{lem:stretchlimit} that $d(u,v) \le w_H(u,v)$, hence also for any $u,v\in V$ we have $d(u,v) \le d_H(u,v)$, where $d_H$ is the shortest distance metric in $H$.

Now, we only need to prove that for every pair of vertices $v,v^*\in X$, we have $d_H(v,v^*) \le C(\eta,\tau) d(v,v^*)$. The proof is similar to that of Lemma~\ref{lem:stretchlimit}.
As before, let $v \in S_1 \in \cS_j$, $v \in S_2 \in \cS_{j+1}$ and $v^* \in S_1^* \in \cS_j$, $v^* \in S_2^* \in \cS_{j+1}$ and assume $S_2$ knows $S_2^*$ but $S_1$ does not know $S_1^*$ (all that is assumed to hold in $\hat{\cT}$, not in $\hat{\cT}(S)$).
Then, since $S_1$ and $S_1^*$ do not know each other, $v$ and $v^*$ are at distance at least $r_j$.
On the other hand, since $S_2$ knows $S_2^*$ in $\hat{\cT}$, we also have that $S_2 \cap S$ knows $S_2^* \cap S$ in $\hat{\cT}(S)$.
Let $u = \opleader(S_2 \cap S)$, $u^* = \opleader(S_2^* \cap S)$. It follows from the definition of $H$, that $uu^*$ is an edge in $H$ and it has weight $w_H(u,u^*) \le \Big( 1 + \frac{4\tau 2^{-\eta}}{\tau - 1}\Big) \tau r_j$.

Now consider the path from ${v}$ to $S_2\cap S$ in $\hat{\cT}(S)$. Initially, $v$ is the leader of the singleton set $\{v\}$, then it might get beaten by some other vertex $v_1$, then $v_1$ can get beaten by some other vertex $v_2$, and so on. Finally, at some level $u$ emerges as a leader. This gives a path $v=v_0,v_1,\ldots,v_m=u$ in $H$. We have
\[ w_H(v_iv_{i+1}) = 2\frac{\tau 2^{-\eta}}{\tau - 1} r_{l_{i+1}} ,\]
 where $l_{i+1}$ is the level at which $v_{i+1}$ beats $v_i$.
 Since all these levels are different and all of them are at most $j+1$, we get:
 \[ d_H(v,u) \le \sum_{i=0}^{m-1} w_H(v_i,v_{i+1}) \le 2\frac{\tau 2^{-\eta}}{\tau - 1} r_0 \sum_{i=0}^j \tau^i
 \le\]
 \[\le 2\frac{\tau 2^{-\eta}}{\tau - 1} \frac{\tau^{j+1}-1}{\tau-1} r_0 \le \frac{2\tau}{\tau-1} \cdot \frac{\tau 2^{-\eta}}{\tau - 1} \tau r_j.\]

 We can argue in the same way for $v^*$ and $u^*$. Joining all 3 bounds we get:
 \[ d_H(v,v^*) \le d_H(v,u)+w_H(u,u^*)+d_H(u^*,v^*) \le\]
 \[\le 
\Big( 1 + \frac{8\tau 2^{-\eta}}{\tau - 1}\Big) \tau r_j + 2 \cdot \frac{2\tau}{\tau-1} \cdot \frac{\tau 2^{-\eta}}{\tau - 1} \tau r_j.\]
and finally
\[ d_H(v,v^*) \le \left( 1 + \left(\frac{\tau}{\tau-1}\right)^2 2^{3-\eta}\right) \tau r_j \le C(\tau,\eta) d(v,v^*).\]
Since every edge of the spanner either corresponds to a father-son edge in $\hat{\cT}(S)$ or to a meeting of two nodes in $\hat{\cT}(S)$, it follows from Lemma~\ref{lem:num_edges} that $H$ has size $O(n)$. The time complexity of constructing $H$ is essentially the same as that of constructing $\hat{\cT}(S)$, i.e.\ $O(k(\log k + \log\log n))$.\qed
\end{proof}





\section{Facility location with unrestricted facilities}
\label{sec:unres-fl}

\ignore{
We have already described that with our tree,
given set of cities $C\subseteq V$ and facilities $F\subseteq V$ with their opening times we can find a $(1.52+\varepsilon)$-approximation of the facility location problem in $\tilde{O}((|C|+|F|)\log\log n)$ time. However, there is a slightly more natural variant of this problem.
Namely, we consider the situation where each point $x$ of the metric $(V, d)$ is assigned an opening cost $f(x)$, i.e.\ the cost of opening the facility in $x$ (we can allow for $f(x)$ to be $\infty$).
Then, the goal is to build a data structure which, given a set of $k$ cities $C\subseteq V$ finds a solution to the facility location problem for cities' set $C$ and facilities' set $V$.
}

In this section we study the variant of {\sc Facility Location} with unrestricted facilities (see Introduction).
We show that our data structure can be augmented to process such queries in $\tilde{O}(k(\log k + \log\log n))$ time,
with the approximation guarantee of $3.04+\varepsilon$.

Our approach here is a reduction to the problem solved in Corollary~\ref{cor:spanner-facility-location}.
The general idea is roughly the following: during the preprocessing phase, for every point $x\in V$ we compute a small set $F(x)$ of facilities that seem a good choice for $x$, and
when processing a query for a set of cities $C$, we just apply Corollary~\ref{cor:spanner-facility-location} to cities' set $C$ and facilities' set $\bigcup_{c\in C}F(c)$.
In what follows we describe the preprocessing and the query algorithm in more detail, and
we analyze the resulting approximation guarantee.

In this section we consider a slightly different representation of tree $\hat{\cT}$.
Namely, we replace each edge $(v,\parent(v))$ of the original $\hat{\cT}$ with a path containing a node for each meeting of $v$.
The nodes on the path are sorted by level, and for any of such nodes $v$, $\level(v)$ denotes the level of the corresponding meeting.
The new tree will be denoted $\bar{\cT}$.

\subsection{Preprocessing}

Let us denote $\vis(j)=\Big( 1 + \frac{4\tau 2^{-\eta}}{\tau - 1}\Big) \tau r_j$, i.e.\ $\vis(j)$ is the upper bound from Lemma~\ref{lem:stretchlimit}.
Note that $\vis(j)$ is an upper bound on the distance between two points $v$ and $w$ such that
$v\in S_1$ and $w\in S_2$ for two sets $S_1, S_2$ that know each other and belong to the same partition $\cS_j$.
For a node $v$ of tree $\cT$ we will also denote $\vis(v)=\vis(\level(v))$.

In the preprocessing, we begin with computing the compressed tree $\bar{\cT}$.
Next, for each node $v$ of $\bar{\cT}$ we compute a point in the sets which $v$ knows,
with the smallest opening cost among these points. Let us denote this point by $\low(v)$.
Finally, for each $x\in V$ consider the path $P$ in $\bar{\cT}$ from the leaf corresponding to $\{x\}$ to the root. Let $P=(v_1, v_2, \ldots, v_{|P|})$ and for $i=1,\ldots,|P|$ let $x_i=\low(v_i)$. Let $p$ the smallest number such that $f(x_p) \le n/\varepsilon_0 \cdot \vis(v_p)$, where $\varepsilon_0$ is a small constant, which we determine later;
now we just assume that $\varepsilon_0\in(0,1]$.
Let $q$ be the smallest number such that $q\ge p$ and $f(x_{q}) \le \varepsilon_0 \cdot \vis(v_{q})$.
If $p$ exists, we let $F(x)=\{v_p, v_{p+1}, \ldots, v_q\}$ and otherwise $F(x)=\emptyset$.

\begin{lemma}
\label{lem:number-of-facilities}
 For any $x\in V$, $|F(x)| = O(\log n)$. \qed
\end{lemma}

\begin{proof}
Let $r=p+ \lceil \log_{\tau}({n}/{\varepsilon_0^2})\rceil$.
Note that for any $i=p, \ldots, r-1$, $\level(v_i)<\level(v_{i+1})$.
Hence
\begin{eqnarray}
\vis(\level(v_p))=\vis(0)\tau^{\level(v_p)}\le \frac{\vis(0)\tau^{\level(v_r)}}{\tau^{r-p}} \le \nonumber \\ \le \frac{\varepsilon_0^2\vis(0)\tau^{\level(v_r)}}{n} = \frac{\varepsilon_0^2}{n} \cdot \vis(\level(v_{r})).\nonumber
\end{eqnarray}
It follows that $f(x_p) \le \varepsilon_0\cdot \vis(\level(v_{r}))$.
Then $q \le r$, since $x_p\in\set(v_r)$.\qed
\end{proof}

It is straightforward to see that all the sets $F(x)$ can be found in $O(n\log n)$ time.

The intuition behind our choice of $F(x)$ is the following.
If $f(x_i) > n/\varepsilon_0 \cdot \vis(v_i)$, then the opening cost of $x_i$ is too high, because
even if $n$ cities contribute to the opening of $x_i$, each of them has to pay more than
$\vis(v_i)$ on average (the constant $\varepsilon_0$ here is needed to deal with some degenerate case, see further), i.e.\ more than an approximation of its connection cost. Hence it is reasonable for cities in $\set(v_i)$ to look a bit further for a cheaper facility.
On the other hand, when $f(x_{i}) \le \varepsilon_0 \cdot \vis(v_{i})$, then even if city $x$ opens facility $x_i$ alone it pays much less than its connection cost to $x_i$. Since the possible cheaper facilities are further than $x_i$, choosing $x_i$ would be a $(1+\varepsilon_0)$-approximation.

\subsection{Query}

Let $C \subseteq V$ be a set of cities passed the query argument.
Denote $k=|C|$. Now for each $c \in C$ we choose the set of facilities
\[F_k(c)= \{\low(v)\ :\ v\in F(c)\textnormal{ and } f(\low(v)) \le k/\varepsilon_0\cdot \vis(v)\}.\]
Similarly as in Lemma~\ref{lem:number-of-facilities} we can show that
$|F_k(c)|=O(\log k)$. Clearly, $F_k(c)$ can be extracted from $F(c)$ in $O(\log k)$ time: if $F(c)$ is sorted w.r.t.\ the level, we just check whether $f(\low(v)) \le k/\varepsilon_0\cdot \vis(v)$ beginning from the highest level vertex and stop when this condition does not hold.
Finally, we compute the union $F(C)=\bigcup_{c\in C} F_{k}(c)\cup\{\low({\rm root}(\bar{\cT})\}$ and we apply Corollary~\ref{cor:spanner-facility-location} to cities' set $C$ and facilities' set $F(C)$. Note that $F$ contains $\low({\rm root}(\bar{\cT})$ -- i.e.\ the point of $V$ with the smallest opening cost --- this is needed to handle some degenerate case.

\subsection{Analysis}

\begin{theorem}
\label{th:fl-reduce}
Let SOL be a solution of the facility location problem for the cities' set $C$ and facilities' set $V$. Then, for any $\varepsilon>0$, there are values of parameters $\tau$, $\eta$ and $\varepsilon_0$ such that there is a solution ${\rm SOL}'$ of cost at most $(2+\varepsilon){\rm cost}({\rm SOL})$, which uses only facilities from set $F(C)$.
\end{theorem}

\begin{proof}
We construct ${\rm SOL}'$ from SOL as follows.
For each opened facility $x$ of SOL, such that $x\not\in F(C)$, we consider the set $C(x)$ of all the cities connected to $x$ in SOL.
We choose a facility $x' \in F(C)$ and reconnect all the cities from $C(x)$ to $x'$.

Let $c^*$ be the city of $C(x)$ which is closest to $x$.
Consider the path $P$ of $\bar{\cT}$ from the leaf corresponding to $c^*$ to the root.
Let $v$ be the first node on this path such that $v$ knows $x$ and
\begin{equation}
\label{eq:f(x)-small}
\vis(v)\ge \frac{\varepsilon_0 f(x)}{|C(x)|}.
\end{equation}
Note that by the first inequality of Lemma~\ref{lem:stretchlimit}, for the first node $w$ on $P$ that knows $x$,
\[\vis(w)\le
\Big( 1 + \frac{4\tau 2^{-\eta}}{\tau - 1}\Big) \tau r_{\level(w)} \le
\Big( 1 + \frac{4\tau 2^{-\eta}}{\tau - 1}\Big) \tau d(x,c^*).\]
On the other hand, again by Lemma~\ref{lem:stretchlimit}, for the first node $u$ on $P$ such that
$\vis(u)\ge \frac{\varepsilon_0 f(x)}{|C(x)|}$, there is $\vis(u)\le \tau\frac{\varepsilon_0 f(x)}{|C(x)|}$.
Hence, since $v$ is the higher of $w$ and $u$,
\begin{equation}
\label{eq:upper-vis(v)}
\vis(v)\le \tau\max\left\{\frac{\varepsilon_0 f(x)}{|C(x)|}, \Big( 1 + \frac{4\tau 2^{-\eta}}{\tau - 1}\Big) d(x,c^*)\right\}.
\end{equation}
First we consider the non-degenerate case when $F_k(c^*)\neq \emptyset$.
Let $v_p, \ldots, v_q$ be the subpath of $P$ which was chosen during the preprocessing.
Let $p'\in \{p,\ldots,q\}$ be the smallest number such that $\low(p') \le k/\varepsilon_0\cdot\vis(v_{p'})$.
Recall that $F_k(c^*)=\{\low(v_i)\ :\ p' \le i \le q\}$.
If $v\in\{v_{p'},\ldots,v_q\}$, then $F_k(c^*)$ contains a facility of opening cost at most $f(x)$, at distance at most $\vis(v)$.
Otherwise $v$ is higher than $v_q$ on $P$, so $F_k(c^*)$ contains a facility of cost at most $\vis(v)$, at distance at most $\varepsilon_0\cdot \vis(v)$.
To sum up, $F_k(c^*)$ contains a facility of cost at most $\max\{f(x),\varepsilon_0\cdot \vis(v)\}$, at distance at most $\vis(v)$. Denote it by $x'$. We reconnect all of $C(x)$ to $x'$.

Now let us bound the cost of connecting $C(x)$ to $x'$. From the triangle
inequality,~\eqref{eq:upper-vis(v)}, and the fact that $c^*$ is closest to $x$ we get
\begin{eqnarray}
\label{eq:conn}
\sum_{c\in C(x)}d(c,x') & \le & \sum_{c\in C(x)}d(c,x) + |C(x)|\vis(v) \nonumber\\
& \le & \sum_{c\in C(x)}d(c,x) + \tau\max\left\{\varepsilon_0 f(x), \Big( 1 + \frac{4\tau 2^{-\eta}}{\tau - 1}\Big) \sum_{c\in C(x)}d(c,x)\right\} \nonumber\\
& \le & \tau \varepsilon_0 f(x) + \left(1+\tau\Big( 1 + \frac{4\tau 2^{-\eta}}{\tau - 1}\Big)\right)\sum_{c\in C(x)}d(c,x).
\end{eqnarray}

Now let us expand the bound for $f(x')$:
\begin{eqnarray}
\label{eq:opening}
f(x') & \le & \max\left\{f(x), \varepsilon_0\tau\max\left\{\frac{\varepsilon_0 f(x)}{|C(x)|}, \Big( 1 + \frac{4\tau 2^{-\eta}}{\tau - 1}\Big) d(x,c^*)\right\}\right\} \nonumber \\
      & \le & (1+\varepsilon_0\tau)f(x) + \varepsilon_0\tau \Big( 1 + \frac{4\tau 2^{-\eta}}{\tau - 1}\Big) \cdot \sum_{c\in C(x)}d(c,x).
\end{eqnarray}

From~\eqref{eq:conn} and~\eqref{eq:opening} together we get
\begin{eqnarray}
\label{eq:non-deg}
 f(x') + &&\sum_{c\in C(x)}d(c,x') \le \\&&\le (1+2\varepsilon_0\tau)f(x)+\left(1+(\tau+\tau\varepsilon_0)\Big( 1 + \frac{4\tau 2^{-\eta}}{\tau - 1}\Big)\right) \sum_{c\in C(x)}d(c,x).\nonumber
\end{eqnarray}

Finally, we handle the degenerate case when $F_k(c^*)=\emptyset$.
Then we just connect all $C(x)$ to the facility $x'=\low({\rm root}(\bar{\cT}))$, i.e.\ the
facility with the smallest opening cost in $V$.
Note that $F_k(c^*)=\emptyset$ implies that for any point $y\in V$ (and hence also for $x$),
\[f(y) > k/\varepsilon_0 \vis ({\rm root}(\bar{\cT})) \ge k/\varepsilon_0 \max_{x,y\in V}d(x,y).\]
Hence, $\sum_{c\in C(x)}d(c,x') \le |C(x)|\max_{x,y\in V}d(x,y) \le (|C(x)|/n)\varepsilon_0 f(x) \le \varepsilon_0 f(x)$. It follows that
\begin{equation}
\label{eq:deg}
f(x') + \sum_{c\in C(x)}d(c,x')  \le (1+\varepsilon_0) f(x).
\end{equation}
From \eqref{eq:non-deg} and \eqref{eq:deg}, we get that
\[{\rm cost}({\rm SOL}') \le \left(1+(\tau+\tau\varepsilon_0)\Big( 1 + \frac{4\tau 2^{-\eta}}{\tau - 1}\Big)\right)  {\rm cost}({\rm SOL}).\]
One can easily seen that the constants $\varepsilon_0$, $\tau$ and $\eta$ can be adjusted so that the coefficient before ${\rm cost}({\rm SOL})$ is arbitrarily close to 2. This proves our claim.\qed
\end{proof}

From Theorem~\ref{th:fl-reduce} and Corollary~\ref{cor:spanner-facility-location} we immediately get the following.

\begin{corollary}[Facility Location with unrestricted facilities, Theorem~\ref{thm:only-cities-facility-location} restated]
\label{cor:only-cities-facility-location}
Assume that for each point of $n$-point $V$ there is assigned an opening cost $f(x)$.
Given $\hat{\cT}$ and a set of $k$ points $C \subseteq V$, for any $\varepsilon>0$, a $(3.04+\varepsilon)$-approximate solution to the facility location problem with cities' set $C$ and facilities' set $V$ can be constructed in time  $O(k\log k (\log^{O(1)}k+\log\log n))$, w.h.p.
\end{corollary}

\section{Dynamic Minimum Spanning Tree and Steiner Tree --- algorithm details}\label{app:dynamic}

In this section we give details on the proof of Theorem \ref{thm:dynamic-mst}, that is,
we describe an algorithm for MST in the static setting and than we make it dynamic.

We assume we have constructed the compressed tree $\hat{\cT}(V)$. Apart from Lemma \ref{lem:who-knows-who},
we treat $\lambda$, $\eta$ and $\tau$ as constants and omit them in the big--$O$ notation.
Recall that we are given a subset $X \subset V$, $|V|=n$, $|X|=k$. In the static
setting, we are to give a constant approximation of a MST for the set $X$ in time
almost linear in $k$. In the dynamic setting, the allowed operations are additions
and removals of vertices to/from $X$ and the goal is achieve polylogarithmic times
on updates.

\subsection{Static Minimum Spanning Tree}
\label{subsec:static_mst}

We first show how the compressed tree $\hat{\cT}(V)$ can be used to solve the static version of the Minimum Spanning Tree problem, i.e.\ we are given a set $X \subseteq V$ and we want to find a constant approximation of the Minimum Spanning Tree of $X$ in time almost linear in $k$.

Let $r$ (called the {\em{root}}) be any fixed vertex in $X$ and let $L_i(X)= \{ x \in X : \textrm{meet}(x,r) = i+1\}$. Also, let $D(\tau) = \left(1+\frac{4\tau 2^{-\eta}}{\tau-1}\right)\tau$.
As a consequence of Property~\ref{item:remstretchlimit} of the partition tree, we get the following:
\begin{lemma}
\label{lem:layer-properties}
Let $x,y \in L_i(X)$. Then $r_i \le d(x,r) \le D(\tau) r_i$. Moreover, at level $i+1+\log_\tau (2D(\tau)) = i+O(1)$ of $\cT$ the sets containing $x$ and $y$ are equal or know each other.
\end{lemma}

A spanning tree $T$ of $X$ is said to be {\em{layered}} if it is a sum of spanning trees for each of the sets $\{r\} \cup L_i(X)$.
The following is very similar to Lemma 8 in Jia et al.~\cite{jia}.

\begin{lemma}
\label{lem:a-la-jia}
There exists a layered tree $T_L$ of $X$ with weight at most $O(1)\OPT$, where $\OPT$ is the weight of the MST of $X$.
\end{lemma}

\begin{proof}
  Let $T_\OPT$ be any minimum spanning tree of $X$ with cost $\OPT$.
  Let $m = \lceil \log_\tau(1 + D(\tau)) \rceil$
  and $\mathcal{L}^j = \{r\} \cup \bigcup \{L_i(X): i \mod m = j\}$ for $0 \leq j < m$.
  Double the edges of $T_\OPT$, walk the resulting graph along its Euler--tour
  and shortcut all vertices not belonging to $\mathcal{L}^j$.
  In this way we obtain the tree $T_\OPT^j$ --- a spanning tree of $\mathcal{L}^j$ with cost at most
  $2\OPT$.
  Clearly $\bigcup_{j=0}^{m-1} T_\OPT^j$ is a spanning tree for $X$ with cost at most $2m\OPT$.

  Let $xy$ be an edge of $T_\OPT^j$ such that $x$ and $y$ belong to different layers, that is,
  $x \in L_i(X)$, $y \in L_{i'}(X)$, $i < i'$.
  Then $i' \geq i + m$ and due to Lemma \ref{lem:stretchlimit}:
  $$d(y, r) \geq r_{i'} \geq (1+D(\tau))r_i \geq r_i + d(x, r).$$
  Therefore $d(x,y) \geq r_i$ and, as $d(x,r) \leq D(\tau)r_i$, $d(x, r) \leq D(\tau)d(x,y)$.
  Moreover:
  $$d(y, r) \leq d(x,y) + d(x,r) \leq d(x,y) + D(\tau)r_i \leq (1+D(\tau))d(x,y).$$
  Therefore by replacing $xy$ by one of the edges $xr$ or $yr$,
  we increase the cost of this edge at most $(1+D(\tau))$ times.
  If we replace all such edges $xy$ in all $T_\OPT^j$ for $0 \leq j < m$,
  we obtain a layered spanning tree of $X$ with cost at most $2m(1+D(\tau))\OPT$.\qed
\end{proof}

Our strategy is to construct an $O(1)$-approximation $T_i(X)$ to MST for each layer separately and connect all these trees to $r$, which gives
an $O(1)$-approximation to MST for $X$ by the above lemma. The spanning tree $T_i(X)$ is constructed using a sparse spanning
graph $G_i(X)$ of $L_i(X)$. In order to build $G_i(X)$ we investigate the levels $\cT$ at which the sets containing vertices of $L_i(X)$
meet each other. The following technical Lemma shows that we can extract this information efficiently:

\begin{lemma}
\label{lem:who-knows-who}
Let $x \in X$ and let $i \le j$ be levels of $\cT$. Then, for each level of $\cT$ in range $[i,j]$, we can find all sets known to $x$,
in total time $O(\log \log \log \metricstretch + \log\log n + \lambda^{\eta+3}(j - i))$.
If we are given level $i_0$ such that there exists a meeting (possibly not involving $x$)
at level $i_0$, the above query takes
$O( (\log \eta + \log \log \lambda)\log \log n + |i_0 - i| + \lambda^{\eta+3}(j-i))$ time.

To perform these queries, the tree needs to be equipped with additional information
of size $O(\lambda^{2\eta + 6}n)$.
\end{lemma}

\begin{proof}
  First we perform $\jump(x, i)$ and, starting with the returned meeting,
  we go up the tree, one meeting at a time, until we reach level $j$.
  In this way, we iterate over all meetings of $x$ between levels $i$ and $j$.
  We store in the tree, for each meeting $(S, S', i')$, the current set of acquaintances
  of $S$ and $S'$. The answer to our query can be retrieved directly from this information.
  Also note that this extra information takes $O(\lambda^{2\eta+6}n)$ space,
  as there are at most $(2n-1)\lambda^{\eta+3}$ meetings and each set knows
  at most $\lambda^{\eta+3}$ other sets at any fixed level.
  If we are given level $i_0$, we may simply perform $\jump(x, i_0)$ and walk
  the tree to level $i$.\qed
\end{proof}

\begin{theorem}\label{thm:static-mst}
Given the compressed tree $\hat{\cT}(V)$ and a subset $X \subseteq V$ one can construct an $O(1)$-approximation $T$ to MST in time
$O(k(\log\log n + \log k)$.
\end{theorem}

\begin{proof}
Designate any $r \in X$ as the root of the tree. Split all the remaining vertices into layers $L_i(X)$.
For each nonempty layer pick a single edge connecting a vertex in this layer to $r$ and add it $T$. Furthermore, add to $T$ an
approximate MST for each layer, denoted $T_i(X)$, constructed as follows.

Consider a layer $L_i(X)$ with $k_i > 0$ elements. We construct a sparse auxiliary graph $G_i(X)$ with $L_i(X)$ as its vertex set.
We use Lemma~\ref{lem:who-knows-who} to find for each vertex $x \in L_i(X)$ and every level in $l \in [i-\log_\tau k+1,i+1+\log_\tau (2D(i))]$
all the sets known to $x$ at level $l$ in $\cT$. Using level $i+1 = \meet(x, r)$ as the anchor $i_0$, this can be done in $O(\log \log n + \log k)$ time per element $x \in L_i(X)$.
Using this information we find, for every $l$ as above and every set $S$ at level
$l$ known to at least one $x \in L_i(X)$, a bucket $B_{l,S}$. This bucket contains all $x \in L_i(X)$ that know $S$ at level $l$. Note that
the total size of the buckets is $O(k_i \log k)$, because we are scanning $O(\log k)$ levels, and a vertex can only know $O(1)$ sets
at each level. Now we are ready to define the edges of $E(G_i(X))$. For every bucket $B_{l,S}$, we add to $E(G_i(X))$ an arbitrary path through
all elements in $B_{l,S}$. We also assign to each edge on this path a weight of $2D(\tau)r_{l-1}$.

Since the total size of the buckets is $O(k_i \log k)$, we also have that $G_i(X)$ has $O(k_i \log k)$ edges. We let $T_i(X)$ be the MST
of $G_i(X)$. Note that $T_i(X)$ can be found in time $O(k_i \log k)$ by the following adjustment of the Kruskal's algorithm.
If we consider buckets ordered in the increasing order of $l$, the edges on the paths are given in the increasing order of their lengths. At every step of Kruskal's algorithm,
we keep current set of connected components of $T_i(X)$ as an array, where every $x \in L_i(X)$ knows the ID of its component. We also keep the size of each component. Whenever two components are joined by a new edge, the new component inherits ID from the bigger subcomponent.
This ensures that for every $x \in L_i(X)$ we change its ID at most $\lceil \log k_i \rceil$ times.

We now claim that
\begin{lemma}
\label{lem:spanning-tree-weight}
The total weight of $T_i(X)$ is $O(1)(\OPT_{L_i(X)} + r_i)$, where $\OPT_{L_i(X)}$ is the weight of the MST for $L_i(X)$.
\end{lemma}

\begin{proof}[Proof of Lemma~\ref{lem:spanning-tree-weight}]
  First, note that the Kruskal's algorithm connects the whole bucket $B_{l,S}$ when considering edges of length $2D(\tau)r_{l-1}$.
  Therefore we may modify graph $G_i(X)$ to $G_i'(X)$ such that all pairs of vertices in $B_{l,S}$ are connected by edges of length $2D(\tau)r_{l-1}$,
  without changing the weight of the MST.
  Let $d_i$ be the metric in $G_i'(X)$. By Lemma \ref{lem:stretchlimit}, $d_i(x,y) \geq d(x,y)$,
  as the sets containing $x$ and $y$ at level, where $x$ and $y$ are not placed in the same bucket,
  do not know each other.
  If $d(x,y) < r_{i-\log_\tau k}$, then
  $d_i(x,y) = 2D(\tau)r_{i-\log_\tau k} = 2D(\tau)r_i/k$,
  as both $x$ and $y$ meet in some bucket at level $i-\log_\tau k+1$.
  Otherwise, if $d(x,y) \geq r_{i-\log_\tau k}$, by Lemma \ref{lem:stretchlimit} again,
  $d_i(x,y) \leq 2D(\tau)d(x,y)$, as both $x$ and $y$ know $S$ at level $l$.

  Let us replace metric $d$ by $d'$ defined as follows: for $x \neq y$, $d'(x,y) = 2D(\tau)d(x,y)$ if $d(x,y) \geq r_{i-\log_\tau k}$
  and $d'(x,y) = r_{i-\log_\tau k}$ otherwise. Clearly $d(x,y) \leq d_i(x,y) \leq d'(x,y)$.
  Note that $d'$ satisfies the following condition: for each $x,y,x',y' \in L_i(X)$, $d(x,y) \leq d(x',y')$ iff $d'(x,y) \leq d'(x',y')$.
  Therefore, the Kruskal's algorithm for MST in $(L_i(X),d)$ chooses the same tree $T_\OPT$ as when run on $(L_i(X),d')$.
  Let $d(T_\OPT)$ ($d_i(T_\OPT)$, $d'(T_\OPT)$) denote the weight of $T_\OPT$ with respect to metric $d$ (resp. $d_i$ and $d'$).
  Let us now bound $d'(T_\OPT)$. $T_\OPT$ consists of $k_i-1$ edges, so the total cost
  of edges $xy$ such that $d(x,y) < r_{i-\log_\tau k}$ is at most $(k_i-1)2D(\tau)\frac{r_i}{k} < 2D(\tau)r_i$.
  Other edges cost at most $D(\tau)$ times more than when using metric $d$, so in total
  $d'(T_\OPT) \leq 2D(\tau)(r_i + d(T_\OPT))$. As $d_i(T_\OPT) \leq d'(T_\OPT)$, and
  Kruskal's algorithm finds minimum MST with respect to $d_i$, the lemma is proven.\qed
\end{proof}

Now we are ready to prove that $T$ is an $O(1)$-approximation of MST for $X$. Let $\OPT$ be the weight of MST for $X$, and $\OPT_L$ be the weight
of the optimal layered MST of $X$. We know that $\OPT_L \le O(1)\OPT$ by Lemma \ref{lem:a-la-jia}.

The optimal solution has to connect $r$ with the vertex $x\in X$ which is the furthest from $r$.
We have $d(r,x) \ge r_{\textrm{max}}$, where
$r_\textrm{max} = r_i$ for the biggest $i$ with non-empty $L_i(X)$. It follows that $\OPT \ge r_{\textrm{max}} \geq O(1)\sum_i r_i$, because the $r_i$-s form a geometric sequence.
Thus, the cost of connecting all layers to $r$ is bounded by $O(1)\OPT$.

Moreover, Lemma \ref{lem:spanning-tree-weight} implies that sum of the weights of all $T_i(X)$-s is bounded by:
\[ O(1)\sum_i \left(r_i + \OPT_{L_i(X)}\right) \leq O(1)\left(\OPT+\OPT_L\right) \le O(1) \OPT.\]
Thus the constructed tree $T$ is an $O(1)$-approximation of the MST of $X$.
\qed
\end{proof}

\subsection{Dynamic Minimum Spanning Tree}
\label{subsec:dynamic_mst}

The dynamic approximation algorithm for MST builds on the ideas from the previous subsection. However, we tackle the following obstacles:
\begin{itemize}
\item we do not have a fixed root vertex around which the layers could be constructed,
\item the number of distance levels considered when building auxiliary graphs is dependent on $k$, and as such can change during the execution of the algorithm, and finally,
\item we need to compute the minimum spanning trees in auxiliary graphs dynamically.
\end{itemize}

The following theorem shows that all of these problems can be solved successfully.

\begin{theorem}[Theorem \ref{thm:dynamic-mst} restated]\label{thm:dynamic-mst-app}
Given the compressed tree $\hat{\cT}(V)$, we can maintain an $O(1)$-approximate Minimum Spanning Tree for a subset $X$ subject to
insertions and deletions of vertices.
The insert operation works in $O(\log^5k + \log\log n)$ time and
the delete operation works in $O(\log^5k)$ time, $k=|X|$.
Both times are expected and amortized.
\end{theorem}

\begin{proof}
  The basic idea is to maintain the layers and the auxiliary graphs described in the proof of Theorem \ref{thm:static-mst}.
  However, since we are not
guaranteed that any vertex is going to permanently stay in $X$, we might need to occasionally recompute the layers and graphs from
scratch, namely when our current root is removed. However, if we always pick root randomly, the probability of this happening as
a result of any given operation is $\le \frac{1}{k}$ and so it will not affect our time bounds. It does, however, make them randomized.

The number of distance levels considered for each layer is $\log_\tau k + O(1)$ and so it might change during the execution of the algorithm.
This can be remedied in many different ways, for example we might recompute all the data structures from scratch every time $k$ changes by a
given constant factor.

The above remarks should make it clear that we can actually maintain the layer structure and the auxiliary graph (as a collection of paths in non--empty buckets $B_{l,S}$) for each layer with low cost (expected and amortized) per update.
We now need to show how to use these structures to dynamically maintain a spanning tree. We use the algorithm of de Lichtenberg, Holm and Thorup (LHT) \cite{hlt:dynamic-mst} that maintains a minimum spanning tree in a graph subject to insertions and deletions of edges, both in time $O(\log^4 n)$, where $n$ is the number of vertices.

We are going to use the LHT algorithm for each auxiliary graph separately. Note that inserting or deleting a vertex corresponds to
inserting or deleting $O(\log k)$ edges to this graph, as every vertex is in $O(\log k)$ buckets.

In case of insertion of a vertex $x$, we need $O(\log \log n)$ time to perform $\meet(v, \textrm{root})$ and find the appropriate layer.
Non--empty layers and their non--empty buckets
may be stored in a dictionary,
so the search for a fixed layer or a fixed bucket is performed in $O(\log k)$ time.
Having appropriate layer, we insert $x$ into all known buckets, taking $O(\log^4 k)$ time to update edges in each bucket.
Therefore the insert operation works in expected and amortized time $O(\log^5 k + \log \log n)$.

In case of deletion of a vertex $x$, we may maintain for each vertex in $X$ a list of its occurrences in buckets, and in this way
we may fast access incident edges. For each occurrence, we delete two incident to $x$ edges and connect the neighbors of $x$.
Therefore the delete operation works in expected and amortized time $O(\log^5 k)$.\qed
\end{proof}

\end{document}